\documentclass{IEEEtran}

\usepackage{amsmath}
\usepackage{amssymb}
\usepackage{graphicx}
\usepackage{verbatim}
\usepackage{cite}
\usepackage{subfigure}
\usepackage{bbm}

\newtheorem{proposition}{Proposition}
\newtheorem{lemma}{Lemma}
\newtheorem{theorem}{Theorem}
\newtheorem{corollary}{Corollary}

\newcommand{\beq}{\begin{equation}}
\newcommand{\eeq}{\end{equation}}
\newcommand{\E}[1]{E \kern -2pt \left\{ #1 \right\}}
\renewcommand{\Pr}[1]{\mathrm{Pr} \kern -2pt \left\{ #1 \right\}}

\DeclareMathOperator{\Var}{Var}
\DeclareMathOperator{\Tr}{Tr}
\DeclareMathOperator{\diag}{diag}

\renewcommand{\b}{{\boldsymbol b}}
\newcommand{\e}{{\boldsymbol e}}
\newcommand{\f}{{\boldsymbol f}}
\newcommand{\h}{{\boldsymbol h}}
\renewcommand{\k}{{\boldsymbol k}}

\renewcommand{\v}{{\boldsymbol v}}
\newcommand{\w}{{\boldsymbol w}}
\newcommand{\x}{{\boldsymbol x}}

\newcommand{\A}{{\boldsymbol A}}

\renewcommand{\H}{{\boldsymbol H}}
\newcommand{\G}{{\boldsymbol G}}
\newcommand{\I}{{\boldsymbol I}}
\newcommand{\J}{{\boldsymbol J}}
\newcommand{\M}{{\boldsymbol M}}

\newcommand{\R}{{\boldsymbol R}}

\newcommand{\U}{{\boldsymbol U}}

\newcommand{\W}{{\boldsymbol W}}

\newcommand{\bmu}{{\boldsymbol \mu}}
\newcommand{\bdel}{{\boldsymbol \delta}}
\newcommand{\bth}{{\boldsymbol \theta}}
\newcommand{\bnu}{{\boldsymbol \nu}}

\newcommand{\hth}{{\hat{\theta}}}
\newcommand{\hbth}{\hat{\bth}}

\newcommand{\RR}{{\mathbb R}}

\newcommand{\MSE}{{\mathrm{MSE}}}

\newcommand{\CRB}{{\mathrm{CRB}}}
\newcommand{\One}{\mathbbm{1}}
\newcommand{\BesselI}[2]{I_{#1} \! \left( {#2} \right)}
\newcommand{\BesselK}[2]{K_{#1} \! \left( {#2} \right)}
\newcommand{\gammainc}[2]{\Gamma_{#1} \! \left( {#2} \right)}
\newcommand{\Qfunc}[1]{Q\!\left(#1\right)}

\newcommand{\bra}{\left \langle }
\newcommand{\ket}{\right \rangle }

\newcommand{\argmin}{\mathop{\arg\min}}
\newcommand{\esssup}{\mathop{\mathrm{ess}\sup}}

\newcommand{\pd}[2]{\frac{\partial #1}{\partial #2}}
\newcommand{\pdd}[3]{\frac{\partial^2 #1}{\partial #2 \partial #3}}

\newcommand{\CoTn}{{C^1}}
\newcommand{\LtTn}{{L^2}}
\newcommand{\HoTn}{{H^1}}
\newcommand{\bol}{{\b_{\text{opt}}^{(\ell)}}}
\newcommand{\bopt}{{\b_{\text{opt}}}}

\makeatletter
\newcommand*\BIG[2]{\mathclose{\hbox{$\left.\vbox to#1{}\right#2\n@space$}}}
\makeatother

\begin{document}

\title{A Lower~Bound on the Bayesian~MSE Based~on~the Optimal Bias Function}

\author{Zvika~Ben-Haim,~\IEEEmembership{Student~Member,~IEEE,}%
and Yonina~C.~Eldar,~\IEEEmembership{Senior~Member,~IEEE}%
\thanks{The authors are with the Department of Electrical Engineering,
Tech\-nion---Is\-rael Institute of Technology, Haifa 32000, Israel
(e-mail: zvikabh@technion.ac.il; yonina@ee.technion.ac.il).
This work was supported in part by the Israel Science Foundation under Grant no. 1081/07 and by the European Commission in the framework of the FP7 Network of Excellence in Wireless COMmunications NEWCOM++ (contract no. 216715).}}

\maketitle

\begin{abstract}
A lower bound on the minimum mean-squared error (MSE) in a Bayesian estimation problem is proposed in this paper. This bound utilizes a well-known connection to the deterministic estimation setting. Using the prior distribution, the bias function which minimizes the Cram\'er--Rao bound can be determined, resulting in a lower bound on the Bayesian MSE\@. The bound is developed for the general case of a vector parameter with an arbitrary probability distribution, and is shown to be asymptotically tight in both the high and low signal-to-noise ratio regimes. A numerical study demonstrates several cases in which the proposed technique is both simpler to compute and tighter than alternative methods.
\end{abstract}


\section{Introduction}

The goal of estimation theory is to infer the value of an unknown parameter based on observations. A common approach to this problem is the Bayesian framework, in which the estimate is constructed by combining the measurements with prior information about the parameter \cite{berger85}. In this setting, the parameter $\bth$ is random, and its distribution describes the \emph{a priori} knowledge of the unknown value. In addition, measurements $\x$ are obtained, whose conditional distribution, given $\bth$, provides further information about the parameter. The objective is to construct an estimator $\hbth$, which is a function of the measurements, so that $\hbth$ is close to $\bth$ in some sense. A common measure of the quality of an estimator is its mean-squared error (MSE), given by $E\{\|\bth-\hbth\|^2\}$.

It is well-known that the posterior mean $\E{\bth|\x}$ is the technique minimizing the MSE\@. Thus, from a theoretical perspective, there is no difficulty in finding the minimum MSE (MMSE) estimator in any given problem. In practice, however, the complexity of computing the posterior mean is often prohibitive. As a result, various alternatives, such as the maximum \emph{a posteriori} (MAP) technique, have been developed \cite{kay93}. The purpose of such methods is to approach the performance of the MMSE estimator with a computationally efficient algorithm.

An important goal is to quantify the performance degradation resulting from the use of these suboptimal techniques. One way to do this is to compare the MSE of the method used in practice with the MMSE\@. Unfortunately, computation of the MMSE is itself infeasible in many cases. This has led to a large body of work seeking to find simple lower bounds on the MMSE in various estimation problems \cite{vantrees68, ZivZakai69, young71, BelliniTartara74, ChazanZakaiZiv75, BobrovskiZakai76, WeissWeinstein85, WeinsteinWeiss88, bell97, renaux06}.

Generally speaking, previous bounds can be divided into two categories. The Weiss--Weinstein family is based on a covariance inequality and includes the Bayesian Cram\'{e}r--Rao bound \cite{vantrees68}, the Bobrovski--Zakai bound \cite{BobrovskiZakai76}, and the Weiss--Weinstein bound \cite{WeissWeinstein85, WeinsteinWeiss88}. The Ziv--Zakai family of bounds is based on comparing the estimation problem to a related detection scenario. This family includes the Ziv--Zakai bound \cite{ZivZakai69} and its improvements, notably the Bellini--Tartara bound \cite{BelliniTartara74}, the Chazan--Zakai--Ziv bound \cite{ChazanZakaiZiv75}, and the generalization of Bell \emph{et al.}\ \cite{bell97}. Recently, Renaux \emph{et~al.} have combined both approaches \cite{renaux06}.

The accuracy of the bounds described above is usually tested numerically in particular estimation settings. Few of the previous results provide any sort of analytical proof of accuracy, even under asymptotic conditions. Bellini and Tartara \cite{BelliniTartara74} briefly discuss performance of their bound at high signal-to-noise ratio (SNR), and Bell \emph{et al.}\ \cite{bell97} prove that their bound converges to the true value at low SNR for a particular family of Gaussian-like probability distributions. To the best of our knowledge, there are no other results concerning the asymptotic performance of Bayesian bounds.

A different estimation setting arises when one considers $\bth$ as a \emph{deterministic} unknown parameter. In this case, too, a common goal is to construct an estimator having low MSE\@. However, the term MSE has a very different meaning in the deterministic setting, since in this case, the expectation is taken only over the random variable $\x$. One elementary difference with far-reaching implications is that in the Bayesian case, the MSE is a single real number, whereas the deterministic MSE is a function of the unknown parameter $\bth$ \cite{lehmann98, eldar08b, kay08}.

Many lower bounds have been developed for the deterministic setting, as well. These include classical results such as the Cram\'er--Rao \cite{cramer45, rao45}, Hammersley--Chapman--Robbins \cite{hammersley50, ChapmanRobbins51}, Bhattacharya \cite{bhattacharya66}, and Barankin \cite{barankin49} bounds, as well as more recent results \cite{abel93, hero96, forster02, eldar04e, eldar06, eldar08d}. By far the simplest and most commonly used of these approaches is the Cram\'er--Rao bound (CRB). Like most other deterministic bounds, the CRB deals explicitly with unbiased estimators, or, equivalently, with estimators having a specific, pre-specified bias function. Two exceptions are the uniform CRB \cite{hero96, eldar04e} and the minimax linear-bias bound \cite{eldar06, eldar08d}. The CRB is known to be \emph{asymptotically} tight in many cases, even though many later bounds are sharper than it \cite{ibragimov81, eldar04e, eldar08b}.

Although the deterministic and Bayesian settings stem from different points of view, there exist insightful relations between the two approaches. The basis for this connection is the fact that by adding a prior distribution for $\bth$, any deterministic problem can be transformed to a corresponding Bayesian setting. Several theorems relate the performance of corresponding Bayesian and deterministic scenarios \cite{lehmann98}. As a consequence, numerous bounds have both a deterministic and a Bayesian version \cite{vantrees68, WeinsteinWeiss88, renaux06, renaux_thesis}.

The simplicity and asymptotic tightness of the deterministic CRB motivate its use in problems in which $\bth$ is random. Such an application was described by Young and Westerberg \cite{young71}, who considered the case of a scalar $\theta$ constrained to the interval $[\theta_0, \theta_1]$. They used the prior distribution of $\theta$ to determine the optimal bias function for use in the biased CRB, and thus obtained a Bayesian bound. It should be noted that this result differs from the Bayesian CRB of Van Trees \cite{vantrees68}; the two bounds are compared in Section~\ref{ss:bayes bound}. We refer to the result of Young and Westerberg as the optimal-bias bound (OBB), since it is based on choosing the bias function which optimizes the CRB using the given prior distribution.

This paper provides an extension and a deeper analysis of the OBB\@. Specifically, we generalize the bound to an arbitrary $n$-dimensional estimation setting \cite{ben-haim07}. The bound is determined by finding the solution to a certain partial differential equation. Using tools from functional analysis, we demonstrate that a unique solution exists for this differential equation. Under suitable symmetry conditions, it is shown that the method can be reduced to the solution of an ordinary differential equation and, in some cases, presented in closed form.

The mathematical tools employed in this paper are also used for characterizing the performance of the OBB\@. Specifically, it is demonstrated analytically that the proposed bound is asymptotically tight for both high and low SNR values. Furthermore, the OBB is compared with several other bounds; in the examples considered, the OBB is both simpler computationally and more accurate than all relevant alternatives.

The remainder of this paper is organized as follows. In Section~\ref{se:bound}, we derive the OBB for a vector parameter. Section~\ref{se:safeguards} discusses some mathematical concepts required to ensure the existence of the OBB\@. In Section~\ref{se:calc}, a practical technique for calculating the bound is developed using variational calculus. In Section~\ref{se:props}, we demonstrate some properties of the OBB, including its asymptotic tightness. Finally, in Section~\ref{se:compare}, we compare the performance of the bound with that of other relevant techniques.

\section{The Optimal-Bias Bound}
\label{se:bound}

In this section, we derive the OBB for the general vector case. To this end, we first examine the relation between the Bayesian and deterministic estimation settings (Section~\ref{ss:bayes-det}). Next, we focus on the deterministic case and review the basic properties of the CRB (Section~\ref{ss:det}). Finally, the OBB is derived from the CRB (Section~\ref{ss:bayes bound}).

The focus of this paper is the Bayesian estimation problem, but the bound we propose stems from the theory of deterministic estimation. To avoid confusion, we will indicate that a particular quantity refers to the deterministic setting by appending the symbol $;\bth$ to it. For example, the notation $\E{\cdot}$ denotes expectation over \emph{both} $\bth$ and $\x$, i.e., expectation in the Bayesian sense, while expectation solely over $\x$ (in the deterministic setting) is denoted by $\E{\cdot ; \bth}$. The notation $\E{\cdot \mid \bth}$ indicates Bayesian expectation conditioned on $\bth$.

Some further notation used throughout the paper is as follows. Lowercase boldface letters signify vectors and uppercase boldface letters indicate matrices. The $i$th component of a vector $\v$ is denoted $v_i$, while $\v^{(1)}, \v^{(2)}, \ldots$ signifies a sequence of vectors. The derivative $\partial f/\partial \v$ of a function $f(\v)$ is a vector function whose $i$th element is $\partial f/\partial v_i$. Similarly, given a vector function $\b(\bth)$, the derivative $\partial \b/\partial \bth$ is defined as the matrix function whose $(i,j)$th entry is $\partial b_i/\partial \theta_j$. The squared Euclidean norm $\v^T \v$ of a vector $\v$ is denoted $\|\v\|^2$, while the squared Frobenius norm $\Tr(\M \M^T)$ of a matrix $\M$ is denoted $\|\M\|_F^2$. In Section~\ref{se:safeguards}, we will also define some functional norms, which will be of use later in the paper.

\subsection{The Bayesian--Deterministic Connection}
\label{ss:bayes-det}

We now review a fundamental relation between the Bayes\-ian and deterministic estimation settings. Let $\bth$ be an unknown random vector in $\RR^n$ and let $\x$ be a measurement vector. The joint probability density function (pdf) of $\bth$ and $\x$ is $p_{\x,\bth}(\x,\bth) = p_{\x|\bth}(\x|\bth) p_\bth(\bth)$, where $p_{\bth}$ is the prior distribution of $\bth$ and $p_{\x|\bth}$ is the conditional distribution of $\x$ given $\bth$. For later use, define the set $\Theta$ of feasible parameter values by
\beq \label{eq:def Theta}
\Theta = \{ \bth \in \RR^n: p_\bth(\bth) > 0 \}.
\eeq
Suppose $\hbth = \hbth(\x)$ is an estimator of $\bth$. Its (Bayesian) MSE is given by
\beq
\MSE = \E{\|\hbth-\bth\|^2} = \int \|\hbth-\bth\|^2 p_{\x,\bth}(\x,\bth) d\x d\bth.
\eeq
By the law of total expectation, we have
\begin{align} \label{eq:total expec}
\MSE
&= \int \left( \int \|\hbth-\bth\|^2 p_{\x|\bth}(\x|\bth) d\x \right) p_\bth(\bth) d\bth
\notag\\
&= \E{ \E{ \|\hbth-\bth\|^2 \Big| \bth } }.
\end{align}

Now consider a deterministic estimation setting, i.e., suppose $\bth$ is a deterministic unknown which is to be estimated from random measurements $\x$. Let the distribution $p_{\x;\bth}$ of $\x$ (as a function of $\bth$) be given by $p_{\x;\bth}(\x;\bth) = p_{\x|\bth}(\x|\bth)$, i.e., the distribution of $\x$ in the deterministic case equals the conditional distribution in the corresponding Bayesian problem.

The estimator $\hbth$ defined above is simply a function of the measurements, and can therefore be applied in the deterministic case as well. Its deterministic MSE is given by \beq \label{eq:det mse}
\E{ \|\hbth-\bth\|^2 ; \bth } = \int \|\hbth-\bth\|^2 p_{\x;\bth}(\x;\bth) d\x
\eeq
Since $p_{\x;\bth}(\x;\bth) = p_{\x|\bth}(\x|\bth)$, we have
\beq
\E{ \|\hbth-\bth\|^2 ; \bth } = \E{ \|\hbth-\bth\|^2 \Big| \bth }.
\eeq
Combining this fact with \eqref{eq:total expec}, we find that the Bayesian MSE equals the expectation of the MSE of the corresponding deterministic problem, i.e.
\beq \label{eq:bayes-det}
\E{ \|\hbth-\bth\|^2 } = \E{ \E{ \|\hbth-\bth\|^2 ; \bth } }.
\eeq
This relation will be used to construct the OBB in Section~\ref{ss:bayes bound}.

\subsection{The Deterministic Cram\'er--Rao Bound}
\label{ss:det}

Before developing the OBB, we review some basic results in the deterministic estimation setting. Suppose $\bth$ is a deterministic parameter vector and let $\x$ be a measurement vector having pdf $p_{\x;\bth}(\x;\bth)$. Denote by $\Theta \subseteq \RR^n$ the set of all possible values of $\bth$. We assume for technical reasons that $\Theta$ is an open set.\footnote{This is required in order to ensure that one can discuss differentiability of $p_{\x;\bth}$ with respect to $\bth$ at any point $\bth \in \Theta$. In the Bayesian setting to which we will return in Section~\ref{ss:bayes bound}, $\Theta$ is defined by \eqref{eq:def Theta}; in this case, adding a boundary to $\Theta$ essentially leaves the setting unchanged, as long as the prior probability for $\bth$ to be on the boundary of $\Theta$ is zero. Therefore, this requirement is of little practical relevance.}

Let $\hbth$ be an estimator of $\bth$ from the measurements $\x$. We require the following regularity conditions to ensure that the CRB holds \cite[\S 3.1.3]{shao03}.

\begin{enumerate}
\item
$p_{\x;\bth}(\x;\bth)$ is continuously differentiable with respect to $\bth$. This condition is required to ensure the existence of the Fisher information.

\item
The Fisher information matrix $\J(\bth)$, defined by
\beq \label{eq:def J}
[\J(\bth)]_{ij} = \E{\pd{\log p_{\x;\bth}}{\theta_i}
                     \pd{\log p_{\x;\bth}}{\theta_j} \, ; \bth}
\eeq
is bounded and positive definite for all $\bth \in \Theta$. This ensures that the measurements contain data about the unknown parameter.

\item
Exchanging the integral and derivative in the equation
\beq \label{eq:CRB reg cond}
  \int t(\x) \pd{}{\theta_i} p_{\x;\theta}(\x;\theta) d\x
= \pd{}{\theta_i} \int t(\x) p_{\x;\theta}(\x;\theta) d\x
\eeq
is justified for any measurable function $t(\x)$, in the sense that, if one side exists, then the other exists and the two sides are equal. A sufficient condition for this to hold is that the support of $p_{\x;\bth}$ does not depend on $\bth$.

\item
All estimators $\hbth$ are Borel measurable functions which satisfy
\beq \label{eq:dom}
\left\| \pd{p_{\x;\bth}}{\bth} \, \hbth^T \right\|_F \le g(\x) \text{ for all }\bth
\eeq
for some integrable function $g(\x)$. This technical requirement is needed in order to exclude certain pathological estimators whose statistical behavior is insufficiently smooth to allow the application of the CRB\@.
\end{enumerate}

The bias of an estimator $\hbth$ is defined as
\beq
\b(\bth) = \E{\hbth;\bth} - \bth.
\eeq
Under the above assumptions, it can be shown that the bias of any estimator is continuously differentiable \cite[Lemma~2]{young71}. Furthermore, under these assumptions, the CRB holds, and thus, for any estimator having bias $\b(\bth)$, we have
\begin{align} \label{eq:CRB}
E&\!\left\{ \|\bth - \hbth\|^2 ; \bth \right\} \ge \CRB[\b,\bth] \notag\\
&\triangleq
\Tr\!\left[\left(\I+\pd{\b}{\bth}\right) \J^{-1}(\bth) \left(\I+\pd{\b}{\bth}\right)^T \right]
+ \|\b(\bth)\|^2.
\end{align}
A more common form of the CRB is obtained by restricting attention to unbiased estimators (i.e., techniques for which $\b(\bth) = {\bf 0}$). Under the unbiasedness assumption, the bound simplifies to $\MSE \ge \Tr(\J^{-1}(\bth))$. However, in the sequel we will make use of the general form \eqref{eq:CRB}.

\subsection{A Bayesian Bound from the CRB}
\label{ss:bayes bound}

The OBB of Young and Westerberg \cite{young71} is based on applying the Bayesian--deterministic connection described in Section~\ref{ss:bayes-det} to the deterministic CRB \eqref{eq:CRB}. Specifically, returning now to the Bayesian setting, one can combine \eqref{eq:bayes-det} and \eqref{eq:CRB} to obtain that, for any estimator $\hbth$ with bias function $\b(\bth)$,
\beq \label{eq:average CRB}
\E{\|\bth - \hbth\|^2} \ge
Z[\b] \triangleq \int_\Theta \CRB[\b,\bth]\, p_\bth(d\bth)
\eeq
where the expectation is now performed over both $\bth$ and $\x$. Note that \eqref{eq:average CRB} describes the Bayesian MSE as a function of a deterministic property (the bias) of $\hbth$. Since any estimator has \emph{some} bias function, and since all bias functions are continuously differentiable in our setting, minimizing $Z[\b]$ over all continuously differentiable functions $\b$ yields a lower bound on the MSE of any Bayesian estimator. Thus, under the regularity conditions of Section~\ref{ss:det}, a lower bound on the Bayesian MSE is given by
\begin{multline} \label{eq:inf}
s = \inf_{\b \in \CoTn} \int_\Theta
\Bigg[ \|\b(\bth)\|^2 + \\
\Tr\!\left(\left(\I+\pd{\b}{\bth}\right) \J^{-1}(\bth) \left(\I+\pd{\b}{\bth}\right)^T \right) \Bigg]
 p_\bth(d\bth)
\end{multline}
where $\CoTn$ is the space of continuously differentiable functions $\f: \Theta \rightarrow {\mathbb R}^n$.

Note that the OBB differs from the Bayesian CRB of Van Trees \cite{vantrees68}. Van Trees' result is based on applying the Cauchy--Schwarz inequality to the joint pdf $p_{\x,\bth}$, whereas the deterministic CRB is based on applying a similar procedure to $p_{\x;\bth}$. As a consequence, the regularity conditions required for the Bayesian CRB are stricter, requiring that $p_{\x,\bth}$ be twice differentiable with respect to $\bth$. By contrast, the OBB requires differentiability only of the conditional pdf $p_{\x|\bth}$. An example in which this difference is important is the case in which the prior distribution $p_\bth$ is discontinuous, e.g., when $p_\bth$ is uniform. The performance of the OBB in this setting will be examined in Section~\ref{se:compare}.

In the next section, we will see that it is advantageous to perform the minimization \eqref{eq:inf} over a somewhat modified class of functions. This will allow us to prove the unique existence of a solution to the optimization problem, a result which will be of use when examining the properties of the bound later in the paper.

\section{Mathematical Safeguards}
\label{se:safeguards}

In the previous section, we saw that a lower bound on the MMSE can be obtained by solving the minimization problem \eqref{eq:inf}. However, at this point, we have no guarantee that the solution $s$ of \eqref{eq:inf} is anywhere near the true value of the MMSE\@. Indeed, at first sight, it may appear that $s=0$ for any estimation setting. To see this, note that $Z[\b]$ is a sum of two components, a bias gradient part and a squared bias part. Both parts are nonnegative, but the former is zero when the bias gradient is $-\I$, while the latter is zero when the bias is zero. No differentiable function $\b$ satisfies these two constraints simultaneously for all $\bth$, since if the squared bias is everywhere zero, then the bias gradient is also zero. However, it is possible to construct a sequence of functions $\b^{(i)}$ for which both the bias gradient and the squared bias norm tend to zero for \emph{almost} every value of $\bth$. An example of such a sequence in a one-dimensional setting is plotted in Fig.~\ref{fig:bias seq}. Here, a sequence $\b^{(i)}$ of smooth, periodic functions is presented. The function period tends to zero, and the percentage of the cycle in which the derivative equals $-1$ increases as $i$ increases. Thus, the pointwise limit of the function sequence is zero almost everywhere, and the pointwise limit of the derivative is $-1$ almost everywhere.

\begin{figure}
\centerline{\includegraphics[scale=0.8]{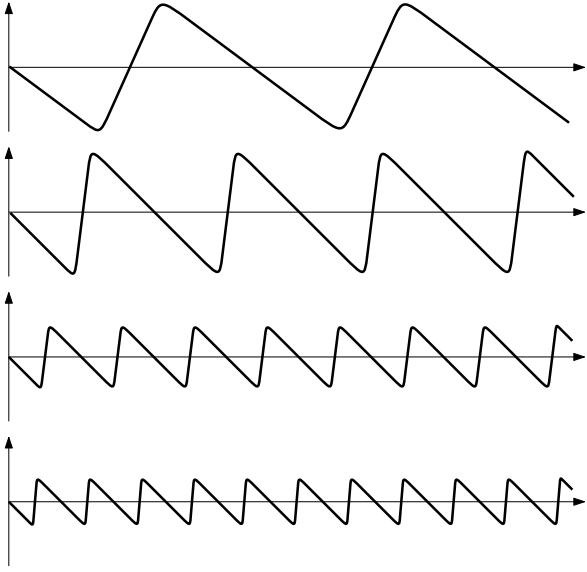}}
\caption{A sequence of continuous functions for which both $|b(\theta)|^2$ and $|1+b'(\theta)|^2$ tend to zero for almost every value of $\theta$.}
\label{fig:bias seq}
\end{figure}

In the specific case shown in Fig.~\ref{fig:bias seq}, it can be shown that the value of $Z[\b^{(i)}]$ does not tend to zero; in fact, $Z[\b^{(i)}]$ tends to infinity in this situation. However, our example illustrates that care must be taken when applying concepts from finite-dimensional optimization problems to variational calculus.

The purpose of this section is to show that $s>0$, so that the bound is meaningful, for any problem setting satisfying the regularity conditions of Section~\ref{ss:det}. (This question was not addressed by Young and Westerberg \cite{young71}.) While doing so, we develop some abstract concepts which will also be used when analyzing the asymptotic properties of the OBB in Section~\ref{se:props}.

As often happens with variational problems, it turns out that the minimum of \eqref{eq:inf} is not necessarily achieved by any continuously differentiable function. In order to guarantee an achievable minimum, one must instead minimize \eqref{eq:inf} over a slightly modified space, which is defined below. As explained in Section~\ref{ss:det}, all bias functions are continuously differentiable, so that the minimizing function ultimately obtained, if it is not differentiable, will not be the bias of any estimator. However, as we will see, the minimum value of our new optimization problem is identical to the infimum of \eqref{eq:inf}. Furthermore, this approach allows us to demonstrate several important theoretical properties of the OBB\@.

Let $\LtTn$ be the space of $p_\bth$-measurable functions $\b: \Theta \rightarrow {\mathbb R}^n$ such that
\beq \label{eq:sq integ}
\int_\Theta \|\b(\bth)\|^2 p_\bth(d\bth) < \infty.
\eeq
Define the associated inner product
\beq
\bra \b^{(1)}, \b^{(2)} \ket_\LtTn \triangleq
\sum_{i=1}^n \int_\Theta b^{(1)}_i(\bth) b^{(2)}_i(\bth) p_\bth(d\bth)
\eeq
and the corresponding norm $\|\b\|_\LtTn^2 \triangleq \bra \b,\b \ket_\LtTn$. Any function $\b \in \LtTn$ has a derivative in the distributional sense, but this derivative might not be a function. For example, discontinuous functions have distributional derivatives which contain a Dirac delta. If, for every $i$, the distributional derivative $\partial b_i/\partial \bth$ of $\b$ is a function in $\LtTn$, then $\b$ is said to be weakly differentiable \cite{lieb01}, and its weak derivative is the matrix function $\partial \b/\partial \bth$. Roughly speaking, a function is weakly differentiable if it is continuous and its derivative exists almost everywhere.

The space of all weakly differentiable functions in $\LtTn$ is called the first-order Sobolev space \cite{lieb01}, and is denoted $\HoTn$. Define an inner product on $\HoTn$ as
\beq \label{eq:HoTn inner prod}
\bra \b^{(1)}, \b^{(2)} \ket_\HoTn \triangleq
\bra \b^{(1)}, \b^{(2)} \ket_\LtTn + \sum_{j=1}^n \bra \pd{b^{(1)}_j}{\bth}, \pd{b^{(2)}_j}{\bth} \ket_\LtTn.
\eeq
The associated norm is $\|\b\|_\HoTn^2 \triangleq \bra \b, \b \ket_\HoTn$. An important property which will be used extensively in our analysis is that $\HoTn$ is a Hilbert space.

Note that since $\Theta$ is an open set, not all functions in $\CoTn$ are in $\HoTn$. For example, in the case $\Theta = \RR^n$, the function $\b(\bth) = \k$, for some nonzero constant $\k$, is continuously differentiable but not integrable. Thus $\b$ is in $\CoTn$ but not in $\HoTn$, nor even in $\LtTn$. However, any measurable function which is not in $\HoTn$ has $\|\b\|_\HoTn = \infty$, meaning that either $\b$ or $\partial \b / \partial \bth$ has infinite $\LtTn$ norm. Consequently, either the bias norm part or the bias gradient part of $Z[\b]$ is infinite. It follows that performing the minimization \eqref{eq:inf} over $\CoTn \cap \HoTn$, rather than over $\CoTn$, does not change the minimum value. On the other hand, $\CoTn \cap \HoTn$ is dense in $\HoTn$, and $Z[\b]$ is continuous, so that minimizing \eqref{eq:inf} over $\HoTn$ rather than $\CoTn \cap \HoTn$ also does not alter the minimum. Consequently, we will henceforth consider the problem
\beq \label{eq:min HoTn}
s = \inf_{\b \in \HoTn} Z[\b].
\eeq

The advantage of including weakly differentiable functions in the minimization is that a unique minimizer can now be guaranteed, as demonstrated by the following result.

\begin{proposition} \label{pr:uniq min}
Consider the problem
\beq
\bar{\b} = \argmin_{\b \in \HoTn} Z[\b]
\eeq
where $Z[\b]$ is given by \eqref{eq:average CRB} and $\J(\bth)$ is positive definite and bounded with probability 1\@. This problem is well-defined, i.e., there exists a unique $\bar{\b} \in \HoTn$ which minimizes $Z[\b]$. Furthermore, the minimum value $s = Z[\bar{\b}]$ is finite and nonzero.
\end{proposition}

Proving the unique existence of a minimizer for \eqref{eq:min HoTn} is a technical exercise in functional analysis which can be found in Appendix~\ref{ap:prf pr:uniq min}. However, once the existence of such a minimizer is demonstrated, it is not difficult to see that $0 < s < \infty$. To see that $s < \infty$, we must find a function $\b$ for which $Z[\b] < \infty$. One such function is $\b = {\bf 0}$, for which $Z[\b]$ is finite since $\J(\bth)$ is bounded. Now suppose by contradiction that $s = 0$, which implies that there exists a function $\bar{\b} \in \HoTn$ such that $Z[\bar{\b}]=0$. Therefore, both the bias gradient and the squared bias parts of $Z[\bar{\b}]$ are zero. In particular, since the squared bias part equals zero, we have $\|\bar{\b}\|_\LtTn = 0$. Hence, $\bar{\b}={\bf 0}$, because $\LtTn$ is a normed space. But then, by the definition \eqref{eq:average CRB} of $Z[\cdot]$,
\beq
Z[\bar{\b}] = \int_\Theta \Tr(\J^{-1}(\bth)) p_\bth(d\bth)
\eeq
which is positive; this is a contradiction.

Note that functions in $\HoTn$ are defined up to changes on a set having zero measure. In particular, the fact that $\b^{(0)}$ is unique does not preclude functions which are identical to $\b^{(0)}$ almost everywhere (which obviously have the same value $Z[\b]$).

Summarizing the discussion of the last two sections, we have the following theorem.

\begin{theorem} \label{th:biased bayesian CRB}
Let $\bth$ be an unknown random vector with pdf $p_\bth(\bth)>0$ over the open set $\Theta \subseteq {\mathbb R}^n$, and let $\x$ be a measurement vector whose pdf, conditioned on $\bth$, is given by $p_{\x|\bth}(\x|\bth)$. Assume the regularity conditions of Section~\ref{ss:det} hold. Then, for any estimator $\hbth$,
\beq \label{eq:th:biased bayesian CRB}
\E{\|\bth-\hbth\|^2} \ge \min_{\b \in \HoTn} \int_\Theta
\CRB[\b,\bth] p_\bth(\bth) d\bth.
\eeq
The minimum in \eqref{eq:th:biased bayesian CRB} is nonzero and finite. Furthermore, this minimum is achieved by a function $\bar{\b} \in \HoTn$, which is unique up to changes having zero probability.
\end{theorem}

Two remarks are in order concerning Theorem~\ref{th:biased bayesian CRB}. First, the function $\b$ solving \eqref{eq:th:biased bayesian CRB} might not be the bias of any estimator; indeed, under our assumptions, all bias functions are continuously differentiable, whereas $\b$ need only be weakly differentiable. Nevertheless, \eqref{eq:th:biased bayesian CRB} is still a lower bound on the MMSE\@. Another important observation is that Theorem~\ref{th:biased bayesian CRB} arises from the deterministic CRB; hence, there are no requirements on the prior distribution $p_\bth(\bth)$. In particular, $p_\bth(\bth)$ can be discontinuous or have bounded support. By contrast, many previous Bayesian bounds do not apply in such circumstances.

\section{Calculating the Bound}
\label{se:calc}

In finite-dimensional convex optimization problems, the requirement of a vanishing first derivative results in a set of equations, whose solution is the global minimum. Analogously, in the case of convex functional optimization problems such as \eqref{eq:th:biased bayesian CRB}, the optimum is given by the solution of a set of  differential equations. The following theorem, whose proof can be found in Appendix~\ref{ap:prf th:e-l}, specifies the differential equation relevant to our optimization problem.

In this section and in the remainder of the paper, we will consider the case in which the set $\Theta = \{ \bth: p_\bth(\bth)>0 \}$ is bounded. From a practical point of view, even when $\Theta$ consists of the entire set $\RR^n$, it can be approximated by a bounded set containing only those values of $\bth$ for which $p_\bth(\bth) > \epsilon$.

\begin{theorem} \label{TH:E-L}
Under the conditions of Theorem~\ref{th:biased bayesian CRB}, suppose $\Theta$ is a bounded subset of ${\mathbb R}^n$ with a smooth boundary $\Lambda$. Then, the optimal $\b(\bth)$ of \eqref{eq:th:biased bayesian CRB} is given by the solution to the system of partial differential equations
\begin{align} \label{eq:th:E-L}
&p_\bth(\bth) b_i(\bth)
= p_\bth(\bth) \sum_{j,k} \pdd{b_i}{\theta_j}{\theta_k} (\J^{-1})_{jk} \notag\\
&\ + \sum_{j,k} \left(\delta_{ik} + \pd{b_i}{\theta_k}\right)
  \left( (\J^{-1})_{jk} \pd{p_\bth}{\theta_j} + p_\bth(\bth) \pd{(\J^{-1})_{jk}}{\theta_j} \right)
\end{align}
for $i=1,\ldots n$, within the range $\bth \in \Theta$, which satisfies the Neumann boundary condition
\beq \label{eq:th:E-L:boundary}
\left(\I + \pd{\b}{\bth}\right) \J^{-1} \bnu(\bth) = {\bf 0}
\eeq
for all points $\bth \in \Lambda$. Here, $\bnu(\bth)$ is a normal to the boundary at $\bth$. All derivatives in this system of equations are to be interpreted in the weak sense.
\end{theorem}

Note that Theorem~\ref{th:biased bayesian CRB} guarantees the existence of a unique solution in $\HoTn$ to the differential equation \eqref{eq:th:E-L} with the boundary conditions \eqref{eq:th:E-L:boundary}.

The bound of Young and Westerberg \cite{young71} is a special case of Theorem~\ref{TH:E-L}, and is given here for completeness.

\begin{corollary} \label{co:E-L 1d}
Under the settings of Theorem~\ref{th:biased bayesian CRB}, suppose $\Theta = (\theta_0,\theta_1)$ is a bounded interval in ${\mathbb R}$. Then, the bias function $b(\theta)$ minimizing \eqref{eq:th:biased bayesian CRB} is a solution to the second-order ordinary differential equation
\beq \label{eq:co:E-L 1d}
J(\theta) b(\theta)
= b''(\theta) + (1+b'(\theta))\left( \frac{d \log p_\bth}{d \theta} - \frac{d \log J}{d \theta} \right)
\eeq
within the range $\theta \in \Theta$, subject to the boundary conditions $b'(\theta_0) = b'(\theta_1) = -1$.
\end{corollary}

Theorem~\ref{TH:E-L} can be solved numerically, thus obtaining a bound for any problem satisfying the regularity conditions. However, directly solving \eqref{eq:th:E-L} becomes increasingly complex as the dimension of the problem increases. Instead, in many cases, symmetry relations in the problem can be used to simplify the solution. As an example, the following spherically symmetric case can be reduced to a one-dimensional setting equivalent to that of Corollary~\ref{co:E-L 1d}. The proof of this theorem can be found in Appendix~\ref{ap:prf th:sph sym}.

\begin{theorem} \label{TH:SPH SYM}
Under the setting of Theorem~\ref{th:biased bayesian CRB}, suppose that $\Theta = \{\bth: \|\bth\| < r \}$ is a sphere centered on the origin, $p_\bth(\bth) = q(\|\bth\|)$ is spherically symmetric, and $\J(\bth) = J(\|\bth\|) \I$, where $J: \RR \rightarrow \RR$ is a scalar function. Then, the optimal-bias bound \eqref{eq:th:biased bayesian CRB} is given by
\begin{align} \label{eq:th:sph sym}
\E{\|\bth-\hbth\|^2}
&\ge \frac{2 \pi^{n/2}}{\Gamma(n/2)}
\int_0^r \Bigg[ b^2(\rho) + \frac{(1 + b'(\rho))^2}{J(\rho)} \notag\\
&\quad {}+ \frac{n-1}{J(\rho)} \left(1 + \frac{b(\rho)}{\rho}\right)^2 \Bigg] q(\rho) \rho^{n-1} d\rho.
\end{align}
Here, $\Gamma(\cdot)$ is the Gamma function, and $b(\rho)$ is a solution to the ODE
\begin{align} \label{eq:th:sph sym ode}
J(\theta) b(\theta)
 = b''(\theta)
 &+ (n-1)\left(\frac{b'(\theta)}{\theta} - \frac{b(\theta)}{\theta^2}\right)\notag\\
 &+ (1+b'(\theta))\left( \frac{d \log q}{d \theta} - \frac{d \log J}{d \theta} \right)
\end{align}
subject to the boundary conditions $b(0)=0$, $b'(r)=-1$. The bias function for which the bound is achieved is given by
\beq
\b(\bth) = b(\|\bth\|) \frac{\bth}{\|\bth\|}.
\eeq
\end{theorem}

In this theorem, the requirement $\J(\bth) = J(\|\bth\|) \I$ indicates that the Fisher information matrix is diagonal and that its components are spherically symmetric. Parameters having a diagonal matrix $\J$ are sometimes referred to as \emph{orthogonal}. The simplest case of orthogonality occurs when, to each parameter $\theta_i$, there corresponds a measurement $x_i$, in such a way that the random variables $x_i|\bth$ are independent. Other orthogonal scenarios can often be constructed by an appropriate parametrization \cite{cox87}.

The requirement that $\J$ have spherically symmetric components occurs, for example, in location problems, i.e., situations in which the measurements have the form $\x = \bth + \w$, where $\w$ is additive noise which is independent of $\bth$. Indeed, under such conditions, $\J$ is constant in $\bth$ \cite[\S 3.1.3]{shao03}. If, in addition, the noise components are independent, then this setting also satisfies the orthogonality requirement, and thus application of Theorem~\ref{TH:SPH SYM} is appropriate. Note that this estimation problem is not separable, since the components of $\bth$ are correlated; thus, the MMSE in this situation is lower than the sum of the components' MMSE\@. An example of such a setting is presented in Section~\ref{se:compare}.

\section{Properties}
\label{se:props}

In this section, we examine several properties of the OBB\@. We first demonstrate that the optimal bias function has zero mean, a property which also characterizes the bias function of the MMSE estimator. Next, we prove that, under very general conditions, the resulting bound is tight at both low and high SNR values. This is an important result, since a desirable property of a Bayesian bound is that it provides an accurate estimate of the ambiguity region between high and low SNR \cite{bell97}. Reliable estimation at the two extremes increases the likelihood that the transition between these two regimes will be correctly identified.

\subsection{Optimal Bias Has Zero Mean}

In any Bayesian estimation problem, the bias of the MMSE estimator $\hbth_{\text{opt}} = \E{\bth|\x}$ has zero mean:
\beq
\E{\hbth_{\text{opt}}} = \E{ \E{\bth|\x} } = \E{\bth}
\eeq
so that
\beq
\E{\b(\hbth_{\text{opt}})} = \E{ \E{\bth|\x} - \bth } = {\bf 0}.
\eeq
Thus, it is interesting to ask whether the optimal bias which minimizes \eqref{eq:th:biased bayesian CRB} also has zero mean. This is indeed the case, as shown by the following theorem.

\begin{theorem}
\label{th:zero mean}
Let $\b(\bth)$ be the solution to \eqref{eq:th:biased bayesian CRB}. Then,
\beq
\E{\b(\bth)} = {\bf 0}.
\eeq
\end{theorem}

\begin{proof}
Assume by contradiction that $\b(\bth)$ has nonzero mean $\E{\b(\bth)} = \bmu \neq {\bf 0}$. Define $\b_0(\bth) \triangleq \b(\bth) - \bmu$. From \eqref{eq:CRB}, we then have
\begin{align}
\CRB[\b_0,\bth] - \CRB[\b,\bth]
 &= \|\b_0(\bth)\|^2 - \|\b(\bth)\|^2 \notag\\
 &= \|\bmu\|^2 - 2 \bmu^T \b(\bth).
\end{align}
Using the functional $Z[\cdot]$ defined in \eqref{eq:average CRB}, we obtain
\begin{align}
Z[\b_0] - Z[\b]
 &= \E{ \|\bmu\|^2 - 2 \bmu^T \b(\bth) } \notag\\
 &= \|\bmu\|^2 - 2 \bmu^T \E{\b(\bth)} \notag\\
 &= - \|\bmu\|^2 < 0.
\end{align}
Thus $Z[\b_0] < Z[\b]$, contradicting the fact that $\b(\bth)$ minimizes \eqref{eq:th:biased bayesian CRB}.
\end{proof}

\subsection{Tightness at Low SNR}

Bell \emph{et al.}\ \cite{bell97} examined the performance of the extended Ziv--Zakai bound at low SNR and demonstrated that, for a particular family of distributions, the extended Ziv--Zakai bound achieves the MSE of the optimal estimator as the SNR tends to $0$. We now examine the low-SNR performance of the OBB, and demonstrate tightness for a much wider range of problem settings.

Bell \emph{et al.}\ did not define the general meaning of a low SNR value, and only stated that ``[a]s observation time and/or SNR become very small, the observations become useless \ldots [and] the minimum MSE estimator converges to the \emph{a priori} mean.'' This statement clearly does not apply to all estimation problems, since it is not always clear what parameter corresponds to the observation time or the SNR\@. We propose to define the zero SNR case more generally as any situation in which $\J(\bth) = {\bf 0}$ with probability 1\@. This definition implies that the measurements do not contain information about the unknown parameter, which is the usual informal meaning of zero SNR\@. In the case $\J(\bth)={\bf 0}$, it can be shown that the MMSE estimator is the prior mean, so that our definition implies the statement of Bell \emph{et al.}

The OBB is inapplicable when $\J(\bth) = {\bf 0}$, since the CRB is based on the assumption that $\J(\bth)$ is positive definite. To avoid this singularity, we consider a sequence of estimation settings which converge to zero SNR\@. More specifically, we require all eigenvalues of $\J(\bth)$ to decrease monotonically to zero for $p_\bth$-almost all $\bth$. The following theorem, the proof of which can be found in Appendix~\ref{ap:asymp}, demonstrates the tightness of the OBB in this low-SNR setting.

\begin{theorem} \label{th:low SNR}
Let $\bth$ be a random vector whose pdf $p_\bth(\bth)$ is nonzero over an open set $\Theta \subseteq {\mathbb R}^n$.
Let $\x^{(1)}, \x^{(2)}, \ldots$ be a sequence of observation vectors having finite Fisher information matrices $\J^{(1)}(\bth), \J^{(2)}(\bth), \ldots$, respectively.
Suppose that, for all $N$, the matrix $\J^{(N)}(\bth)$ is positive definite for $p_\bth$-almost all $\bth$, and that all eigenvalues of $\J^{(N)}(\bth)$ decrease monotonically to zero as $N \rightarrow \infty$ for $p_\bth$-almost all $\bth$.
Let $\beta_N$ denote the optimal-bias bound for estimating $\bth$ from $\x^{(N)}$. Then,
\beq
\lim_{N \rightarrow \infty} \beta_N = \E{ \left\| \bth - \E{\bth} \right\|^2 }.
\eeq
\end{theorem}

\subsection{Tightness at High SNR}

We now examine the performance of the OBB for high SNR values. To formally define the high SNR regime, we consider a sequence of measurements $\x^{(1)}, \x^{(2)}, \ldots$ of a single parameter vector $\bth$. It is assumed that, when conditioned on $\bth$, all measurements $\x^{(i)}$ are identically and independently distributed (IID). Furthermore, we assume that the Fisher information matrix of a single observation $\J(\bth)$ is well-defined, positive definite and finite for $p_\bth$-almost all $\bth$. We consider the problem of estimating $\bth$ from the set of measurements $\{ \x^{(1)}, \ldots, \x^{(N)} \}$, for a given value of $N$. The high SNR regime is obtained when $N$ is large.

When $N$ tends to infinity, the MSE of the optimal estimator tends to zero. An important question, however, concerns the rate of convergence of the minimum MSE\@. More precisely, given the optimal estimator $\hbth^{(N)}$ of $\bth$ from $\{ \x^{(1)}, \ldots, \x^{(N)} \}$, one would like to determine the asymptotic distribution of $\sqrt{N} (\hbth^{(N)} - \bth)$, conditioned on $\bth$. A fundamental result of asymptotic estimation theory can be loosely stated as follows \cite[\S III.3]{ibragimov81}, \cite[\S 6.8]{lehmann98}. Under some fairly mild regularity conditions, the asymptotic distribution of $\sqrt{N} (\hbth^{(N)}-\bth)$, conditioned on $\bth$, does not depend on the prior distribution $p_\bth$; rather, $\sqrt{N} (\hbth^{(N)}-\bth) \, | \, \bth$ converges in distribution to a Gaussian random vector with mean zero and covariance $\J^{-1}(\bth)$. It follows that
\beq \label{eq:asymptotic MSE}
\lim_{N \rightarrow \infty} N \E{ \|\hbth^{(N)} - \bth\|^2 } = \E{\Tr[\J^{-1}(\bth)]}.
\eeq

Since the minimum MSE tends to zero at high SNR, any lower bound on the minimum MSE must also tend to zero as $N \rightarrow \infty$. However, one would further expect a good lower bound to follow the behavior of \eqref{eq:asymptotic MSE}. In other words, if $\beta_N$ represents the lower bound for estimating $\bth$ from $\{ \x^{(1)},\ldots,\x^{(N)} \}$, a desirable property is $N \beta_N \rightarrow \E{\Tr[\J^{-1}(\bth)]}$. The following theorem, whose proof is found in Appendix~\ref{ap:asymp}, demonstrates that this is indeed the case for the OBB.

Except for a very brief treatment by Bellini and Tartara \cite{BelliniTartara74}, no previous Bayesian bound has shown such a result. Although it appears that the Ziv--Zakai and Weiss--Weinstein bounds may also satisfy this property, this has not been proven formally. It is also known that the Bayesian CRB is \emph{not} asymptotically tight in this sense \cite[Eqs.~(37)--(39)]{vantrees07}.

\begin{theorem}
\label{th:high SNR}
Let $\bth$ be a random vector whose pdf $p_\bth(\bth)$ is nonzero over an open set $\Theta \subseteq {\mathbb R}^n$.
Let $\x^{(1)}, \x^{(2)}, \ldots$ be a sequence of measurement vectors, such that $\x^{(1)}|\bth, \x^{(2)}|\bth, \ldots$ are IID.
Let $\J(\bth)$ be the Fisher information matrix for estimating $\bth$ from $\x^{(1)}$, and suppose $\J(\bth)$ is finite and positive definite for $p_\bth$-almost all $\bth$.
Let $\beta_N$ be the optimal-bias bound \eqref{eq:th:biased bayesian CRB} for estimating $\bth$ from the observation sequence $\{ \x^{(1)}, \ldots, \x^{(N)} \}$. Then,
\beq \label{eq:th:high SNR}
\lim_{N \rightarrow \infty} N \beta_N = \E{\Tr(\J^{-1}(\bth))}.
\eeq
\end{theorem}

Note that for Theorem~\ref{th:high SNR} to hold, we require only that $\J(\bth)$ be finite and positive definite. By contrast, the various theorems guaranteeing asymptotic efficiency of Bayesian estimators all require substantially stronger regularity conditions \cite[\S III.3]{ibragimov81}, \cite[\S 6.8]{lehmann98}. One reason for this is that asymptotic efficiency describes the behavior of $\hbth$ conditioned on each possible value of $\bth$, and is thus a stronger result than the asymptotic Bayesian MSE of \eqref{eq:asymptotic MSE}.

\section{Example: Uniform Prior}
\label{se:compare}

The original bound of Young and Westerberg \cite{young71} predates most Bayesian bounds, and, surprisingly, it has never been cited by or compared with later results. In this section, we measure the performance of the original bound and of its extension to the vector case against that of various other techniques. We consider the case in which $\bth$ is uniformly distributed over an $n$-dimensional open ball $\Theta = \{ \bth: \|\bth\| < r \} \subseteq {\mathbb R}^n$, so that
\beq
p_\bth(\bth) = \frac{1}{V_n(r)} \One_\Theta
\eeq
where $\One_S$ equals $1$ when $\bth \in S$ and $0$ otherwise, and
\beq
V_n(r) = \frac{\pi^{n/2} r^{n-1}}{\Gamma(1+n/2)}
\eeq
is the volume of an $n$-ball of radius $r$ \cite{vinogradov95}. We further assume that
\beq
\x = \bth + \w
\eeq
where $\w$ is zero-mean Gaussian noise, independent of $\bth$, having covariance $\sigma^2 \I$. We are interested in lower bounds on the MSE achievable by an estimator of $\bth$ from $\x$.

We begin by developing the OBB for this setting, as well as some alternative bounds. We then compare the different approaches in a one-dimensional and a three-dimensional setting.

The Fisher information matrix for the given estimation problem is given by $\J(\bth) = \sigma^{-2} \I$, so that the conditions of Theorem~\ref{TH:SPH SYM} hold. It follows that the optimal bias function is given by $\b(\bth) = b(\|\bth\|) \bth/\|\bth\|$, where $b(\cdot)$ is a solution to the differential equation
\beq
\frac{b}{\sigma^2} = b'' + (n-1) \left( \frac{b'}{\theta} - \frac{b}{\theta^2} \right)
\eeq
with boundary conditions $b(0)=0$, $b'(r)=-1$. The general solution to this differential equation is given by
\beq \label{eq:pde sol}
b(\theta) = C_1 \theta^{1 - n/2} \BesselI{n/2}{\frac{\theta}{\sigma}}
          + C_2 \theta^{1 - n/2} \BesselK{n/2}{\frac{\theta}{\sigma}}
\eeq
where $\BesselI{\alpha}{z}$ and $\BesselK{\alpha}{z}$ are the modified Bessel functions of the first and second types, respectively \cite{abramowitz64}. Since $\BesselK{\alpha}{z}$ is singular at the origin, the requirement $b(0)=0$ leads to $C_2=0$. Differentiating \eqref{eq:pde sol} with respect to $\theta$, we obtain
\beq
b'(\theta) = C_1 \theta^{-n/2} \left( \BesselI{n/2}{\frac{\theta}{\sigma}}
            + \frac{\theta}{\sigma} \BesselI{1+n/2}{\frac{\theta}{\sigma}} \right)
\eeq
so that the requirement $b'(r)=-1$ leads to
\beq
C_1 = - \frac{r^{n/2}}{\BesselI{n/2}{r/\sigma} + r/\sigma \BesselI{1+n/2}{r/\sigma}}.
\eeq
Substituting this value of $b(\cdot)$ into \eqref{eq:th:sph sym} yields the OBB, which can be computed by evaluating a single one-dimensional integral. Alternatively, in the one-dimensional case, the integral can be computed analytically, as will be shown below.

Despite the widespread use of finite-support prior distributions \cite{WeinsteinWeiss88, ZivZakai69}, the regularity conditions of many bounds are violated by such prior pdf functions. Indeed, the Bayesian CRB of Van Trees \cite{vantrees68}, the Bobrovski--Zakai bound \cite{BobrovskiZakai76}, and the Bayesian Abel bound \cite{renaux06} all assume that $p_\bth(\bth)$ has infinite support, and thus cannot be applied in this scenario.

Techniques from the Ziv--Zakai family are applicable to constrained problems. An extension of the Ziv--Zakai bound for vector parameter estimation was developed by Bell \emph{et al.}\ \cite{bell97}. From \cite[Property~4]{bell97}, the MSE of the $i$th component of $\bth$ is bounded by
\beq \label{eq:ezzb1}
\E{(\theta_i - \hat{\theta}_i)^2}
 \ge \int_0^\infty V \left\{ \max_{\bdel: \e_i^T \bdel = h} A(\bdel) P_{\min}(\bdel) \right\} h \, d h
\eeq
where $\e_i$ is a unit vector in the direction of the $i$th component, $V\{\cdot\}$ is the valley-filling function defined by
\beq
V\{ f(h) \} = \max_{\eta \ge 0} f(h+\eta),
\eeq
\beq
A(\bdel) \triangleq
 \int_{{\mathbb R}^n} \min \left( p_\bth(\bth), p_\bth(\bth+\bdel) \right) d\bth,
\eeq
and $P_{\min}(\bdel)$ is the minimum probability of error for the problem of testing hypothesis $H_0: \bth = \bth_0$ vs.\ $H_1: \bth = \bth_0 + \bdel$. In the current setting, $P_{\min}(\bdel)$ is given by $P_{\min}(\bdel) = \Qfunc{\|\bdel\|/2\sigma}$, where $\Qfunc{z} = (2\pi)^{-1/2} \int_z^{\infty} e^{-t^2/2} dt$ is the tail function of the normal distribution. Also, we have
\beq
A(\bdel) = \frac{V_n^C(r,\|\bdel\|)}{V_n(r)}
\eeq
where
\beq
V_n^C(r,h) = \int_{{\mathbb R}^n} \One_\Theta \One_{\Theta+h \e_1} d\bth
\eeq
and $\Theta + h \e_1 = \{ \bth + h \e_1 : \bth \in \Theta \}$. Thus, $V_n^C(r,h)$ is the volume of the intersection of two $n$-balls whose centers are at a distance of $h$ units from one another. Substituting these results into \eqref{eq:ezzb1}, we have
\begin{align} \label{eq:ezzb2}
&\E{(\theta_i - \hat{\theta}_i)^2} \notag\\
&\quad \ge \int_0^\infty V \left\{ \max_{\bdel: \e_i^T \bdel = h}  \frac{V_n^C(r,\|\bdel\|)}{V_n(r)} \Qfunc{\frac{\|\bdel\|}{2\sigma}} \right\} h \, d h.
\end{align}
Note that both $V_n^C(r,\|\bdel\|)$ and $Q(\|\bdel\|/2\sigma)$ decrease with $\|\bdel\|$. Therefore, the maximum in \eqref{eq:ezzb2} is obtained for $\bdel = h \e_i$. Also, since the argument of $V\{\cdot\}$ is monotonically decreasing, the valley-filling function has no effect and can be removed. Finally, since $V_n^C(r,h)=0$ for $h>2r$, the integration can be limited to the range $[0,2r]$. Thus, the extended Ziv--Zakai bound is given by
\beq \label{eq:ezzb f}
\E{\|\bth - \hbth\|^2}
 \ge \int_0^{2r} n \frac{V_n^C(r,h)}{V_n(r)} \Qfunc{\frac{h}{2\sigma}}  h \, d h.
\eeq

\begin{figure*}
\centerline{%
\subfigure[]{%
\includegraphics{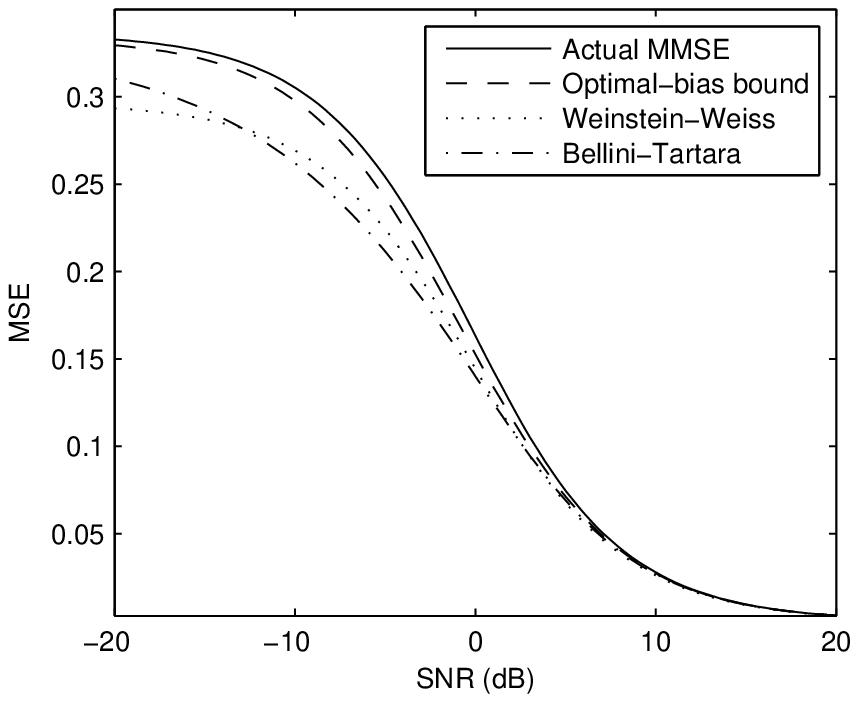}%
\label{fig:1d mse}}
\hfil
\subfigure[]{%
\includegraphics{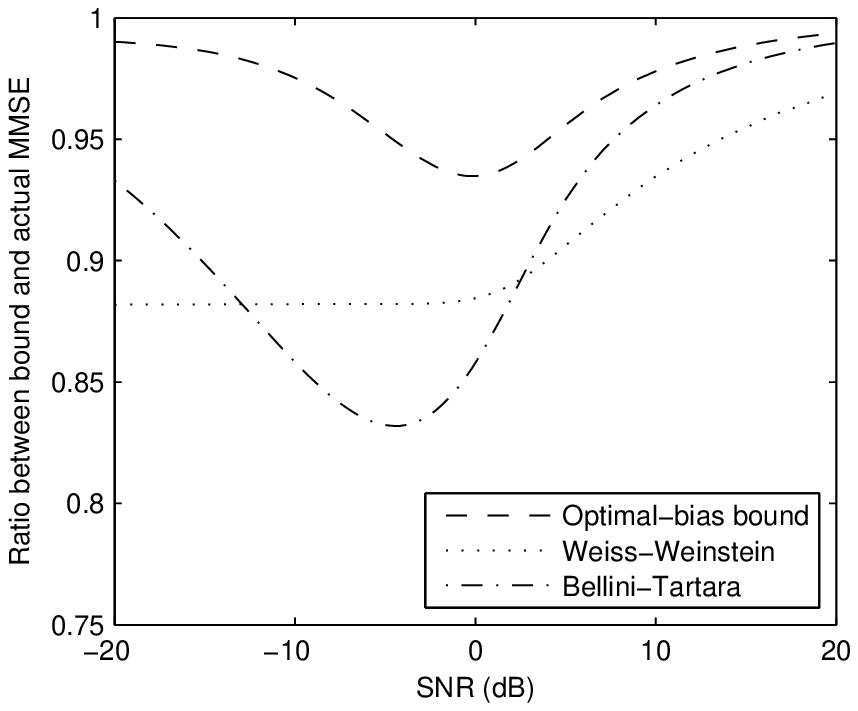}%
\label{fig:1d ratio}}}
\caption{Comparison of the MSE bounds and the minimum achievable MSE in a one-dimensional setting for which $\theta \sim U[-r,r]$ and $x|\theta \sim N(\theta, \sigma^2)$.}
\label{fig:1d}
\end{figure*}

We now compute the Weiss--Weinstein bound for the setting at hand. This bound is given by
\beq \label{eq:wwb}
\E{\|\bth - \hbth\|^2} \ge \Tr(\H \G^{-1} \H^T)
\eeq
where $\H = [\h_1, \ldots, \h_m]$ is a matrix containing an arbitrary number $m$ of test vectors and $\G$ is a matrix whose elements are given by
\beq
G_{ij} = \frac{\E{r(\x,\bth;\h_i,s_i) r(\x,\bth;\h_j,s_j)}}
              {\E{L^{s_i}(\x;\bth+\h_i,\bth)}\E{L^{s_j}(\x;\bth+\h_j,\bth)}}
\eeq
in which
\beq
r(\x,\bth;\h_i,s_i) \triangleq L^{s_i}(\x;\bth+\h_i,\bth) - L^{1-s_i}(\x;\bth-\h_i,\bth)
\eeq
and
\beq
L(\x;\bth_1,\bth_2) \triangleq \frac{p_\bth(\bth_1) p_{\x|\bth}(\x|\bth_1)}{p_\bth(\bth_2) p_{\x|\bth}(\x|\bth_2)}.
\eeq
The vectors $\h_1, \ldots, \h_m$ and the scalars $s_1, \ldots, s_m$ are arbitrary, and can be optimized to maximize the bound \eqref{eq:wwb}. To avoid a multidimensional nonconvex optimization problem, we restrict attention to $m=n$, $\h_i = h \e_i$, and $s_i = 1/2$, as suggested by \cite{WeinsteinWeiss88}. This results in a dependency on a single scalar parameter $h$.

Under these conditions, $G_{ij}$ can be written as
\begin{multline} \label{eq:Gij}
G_{ij} =
\frac{1}{M(\h_i) M(\h_j)}
  \big[
    \tilde{M}(\h_i-\h_j,-\h_j) + \tilde{M}(\h_i-\h_j,\h_i) \\
  - \tilde{M}(\h_i+\h_j,\h_j) - \tilde{M}(\h_i+\h_j,\h_i)
  \big]
\end{multline}
where
\beq
M(\h) \triangleq \E{L^{1/2}(\x;\bth+\h,\bth)}
\eeq
and
\beq
\tilde{M}(\h_1,\h_2) \triangleq \E{L^{1/2}(\x;\bth+\h_1,\bth) \One_{\Theta+\h_2}}.
\eeq
Note that we have used the corrected version of the Weiss--Weinstein bound \cite{ben-haim07b}. Substituting the probability distribution of $\x$ and $\bth$ into the definitions of $M(\h)$ and $\tilde{M}(\h_1,\h_2)$, we have
\begin{align}
M(\h)
&= \E{e^{-\|\bth+\h-\x\|^2/4\sigma^2} e^{\|\bth-\x\|^2/4\sigma^2} \One_{\Theta+\h}} \notag\\
&= \frac{V_n^C(r,\|\h\|)}{V_n(r)}  e^{-\|\h\|^2/8\sigma^2}
\end{align}
and, similarly,
\begin{align} \label{eq:tilde M}
\tilde{M}(\h_1,\h_2)
&= \frac{e^{-\|\h_1\|^2/8\sigma^2}}{V_n(r)} \int \One_\Theta \One_{\Theta+\h_1} \One_{\Theta+\h_2} d\bth .
\end{align}
Thus, $M(\h)$ is a function only of $\|\h\|$, and $\tilde{M}(\h_1,\h_2)$ is a function only of $\|\h_1\|$, $\|\h_2\|$, and $\|\h_1-\h_2\|$. Since $\h_i = h \e_i$, it follows that, for $i \neq j$, the numerator of \eqref{eq:Gij} vanishes. Thus, $\G$ is a diagonal matrix, whose diagonal elements equal
\beq
G_{ii} = 2 \frac{ \tilde{M}(0,h \e_1) - \tilde{M}(2h\e_1, h\e_1) }{ M^2(h \e_1) }.
\eeq
The Weiss--Weinstein bound is given by substituting this result into \eqref{eq:wwb} and maximizing over $h$, i.e.,
\beq \label{eq:WWB vec}
\E{\|\bth - \hbth\|^2}
 \ge \max_{h \in [0,2r]} \frac{n h^2 M^2(h \e_1)}{2[\tilde{M}(0,h \e_1) - \tilde{M}(2h\e_1,h\e_1)]}.
\eeq
The value of $h$ yielding the tightest bound can be determined by performing a grid search.

To compare the OBB with the alternative approaches developed above, we first consider the one-dimensional case in which $\theta$ is uniformly distributed in the range $\Theta = (-r,r)$. Let $x = \theta + w$ be a single noisy observation, where $w$ is zero-mean Gaussian noise, independent of $\theta$, with variance $\sigma^2$. We wish to bound the MSE of an estimator of $\theta$ from $x$.

The optimal bias function is given by \eqref{eq:pde sol}. Using the fact that $\BesselI{1/2}{t} = \sqrt{2/\pi} \sinh(t)/\sqrt{t}$, we obtain
\beq
b(\theta) = -\sigma \frac{\sinh(\theta/\sigma)}{\cosh(r/\sigma)}
\eeq
which also follows \cite{young71} from Corollary~\ref{co:E-L 1d}. Substituting this expression into \eqref{eq:th:biased bayesian CRB}, we have that, for any estimator $\hth$,
\beq \label{eq:bound unif gauss}
\E{(\theta - \hth)^2} \ge \sigma^2\left( 1 - \frac{\tanh(r/\sigma)}{r/\sigma} \right).
\eeq

Apart from the reduction in computational complexity, the simplicity of \eqref{eq:bound unif gauss} also emphasizes several features of the estimation problem. First, the dependence of the problem on the dimensionless quantity $r/\sigma$, rather than on $r$ and $\sigma$ separately, is clear. This is to be expected, as a change in units of measurement would multiply both $r$ and $\sigma$ by a constant. Second, the asymptotic properties demonstrated in Theorems \ref{th:low SNR} and \ref{th:high SNR} can be easily verified. For $r \gg \sigma$, the bound converges to the noise variance $\sigma^2$, corresponding to an uninformative prior whose optimal estimator is $\hth = x$; whereas, for $\sigma \gg r$, a Taylor expansion of $\tanh(z)/z$ immediately shows that the bound converges to $r^2/3$, corresponding to the case of uninformative measurements, where the optimal estimator is $\hth = 0$. Thus, the bound \eqref{eq:bound unif gauss} is tight both for very low and for very high SNR, as expected.

In the one-dimensional case, we have $V_1(r) = 2r$ and $V_1^C(r,h) = \max(2r-h,0)$, so that the extended Ziv--Zakai bound \eqref{eq:ezzb f} and the Weiss--Weinstein bound \eqref{eq:WWB vec} can also be simplified somewhat. In particular, the extended Ziv--Zakai bound \eqref{eq:ezzb f} can be written as
\beq \label{eq:ezzb 1d}
\E{\|\bth - \hbth\|^2}
 \ge \int_0^{2r} \left( 1 - \frac{h}{2r} \right) h \Qfunc{\frac{h}{2\sigma}}  d h.
\eeq
Using integration by parts, \eqref{eq:ezzb 1d} becomes
\begin{multline} \label{eq:ezzb 1d final}
\E{\|\bth - \hbth\|^2}
 \ge \frac{2r^2}{3} \Qfunc{\frac{r}{\sigma}} \\
  + \sigma^2 \left[ \gammainc{3/2}{\frac{r^2}{2\sigma^2}}
                        - \frac{8}{3 \sqrt{2\pi}} \frac{\sigma}{r}
                             \gammainc{2}{\frac{r^2}{2\sigma^2}}
             \right]
\end{multline}
where $\gammainc{a}{z} = (1/\Gamma(a)) \int_0^z e^{-t} t^{a-1} dt$ is the incomplete Gamma function. Like the expression \eqref{eq:bound unif gauss} for the OBB, this bound can be shown to converge to the noise variance $\sigma^2$ when $r \gg \sigma$ and to the prior variance $r^2/3$ when $\sigma \gg r$. However, while the convergence of the OBB to these asymptotic values has been demonstrated in general in Theorems \ref{th:low SNR} and~\ref{th:high SNR}, the asymptotic tightness of the Ziv--Zakai bound in the general case remains an open question.

The Weiss--Weinstein bound \eqref{eq:WWB vec} can likewise be simplified further in the one-dimensional case, yielding
\begin{align}
&\E{\|\bth - \hbth\|^2}\notag\\
&\quad \ge \max_{h \in [0,2r]}
  \frac{ h^2 e^{-h^2/4\sigma^2} \left( 1 - \frac{h}{2r} \right)^2 }
       { 2 \left(
          1 - \frac{h}{2r} - \max\left(0, 1-\frac{h}{r}\right) e^{-h^2/2\sigma^2}
       \right) }.
\end{align}
However, calculating this bound still requires a numerical search for the optimal value of $h$.

\begin{figure*}
\centerline{%
\subfigure[]{%
\includegraphics{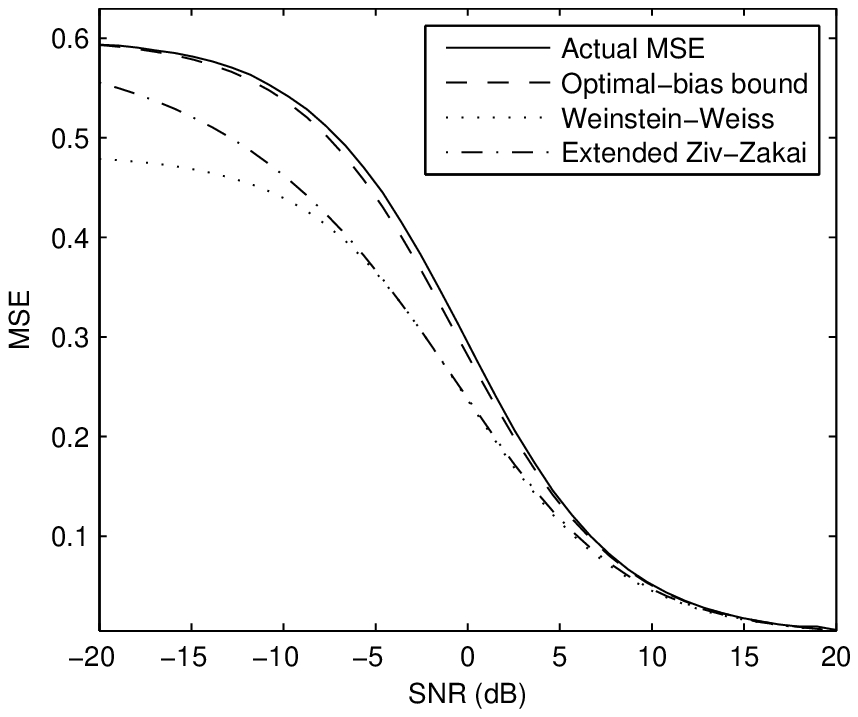}%
\label{fig:3d mse}}
\hfil
\subfigure[]{%
\includegraphics{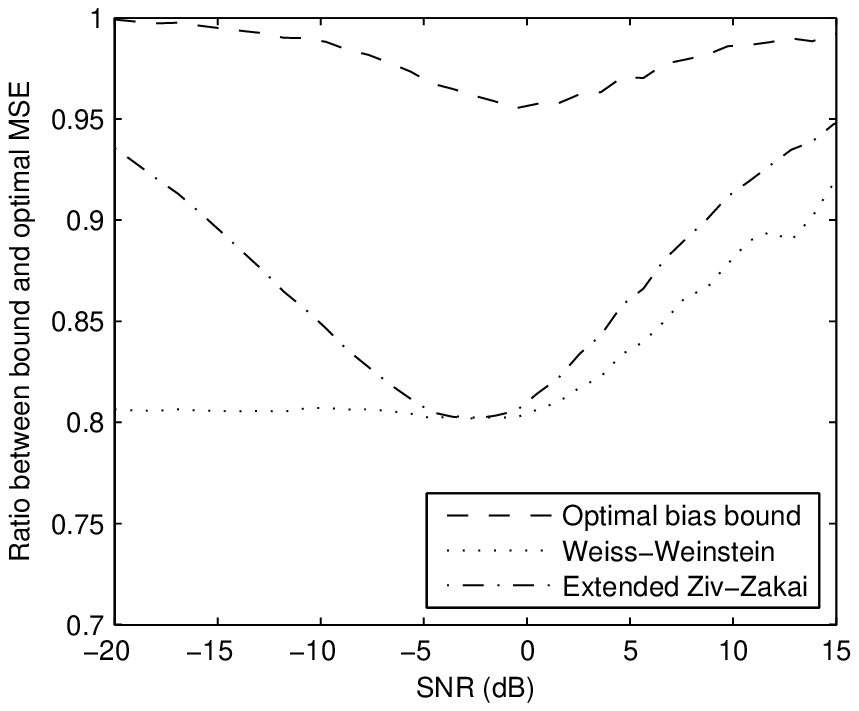}%
\label{fig:3d ratio}}}
\caption{Comparison of the MSE bounds and the minimum achievable MSE in a three-dimensional setting for which $\bth$ is uniformly distributed over a ball of radius $r$ and $\x|\bth \sim N(\bth, \sigma^2\I)$.}
\label{fig:3d}
\end{figure*}

These bounds are compared with the exact value of the MMSE in Fig.~\ref{fig:1d}. In this figure, the SNR is defined as
\beq \label{eq:WWB 1d}
\mathrm{SNR (dB)}
 = 10 \log_{10} \left( \frac{\Var(\theta)}{\Var(w)} \right)
 = 10 \log_{10} \left( \frac{r^2}{3\sigma^2} \right).
\eeq
The MMSE was computed by Monte Carlo approximation of the error of the optimal estimator $E\{\theta|x\}$, which was itself computed by numerical integration. Fig.~\ref{fig:1d mse} plots the MMSE and the values obtained by the aforementioned bounds, while Fig.~\ref{fig:1d ratio} plots the ratio between each of the bounds and the actual MMSE in order to emphasize the difference in accuracy between the various bounds. As can be seen from this figure, the OBB is closer to the true MSE than all other bounds, for all tested SNR values.

The improvements provided by the OBB continue to hold in higher dimensions as well, although in this case it is not possible to provide a closed form for any of the bounds. For example, Fig.~\ref{fig:3d} compares the aforementioned bounds with the true MMSE in the three-dimensional case. In this case the SNR is given by
\beq
\mathrm{SNR (dB)}
 = 10 \log_{10} \left( \frac{\Var(\bth)}{\Var(\w)} \right)
 = 10 \log_{10} \left( \frac{r^2}{5\sigma^2} \right).
\eeq
Here, computation of the minimum MSE requires multi-dimensional numerical integration, and is by far more computationally complex than the calculation of the bounds. Again, it is evident from this figure that the OBB is a very tight bound in all ranges of operation, and is considerably closer to the true value than either of the alternative approaches.

\section{Conclusion}

Although often considered distinct settings, there are insightful connections between the Bayesian and deterministic estimation problems. One such relation is the use of the deterministic CRB in a Bayesian problem. The application of this deterministic bound to the problem of estimating the minimum Bayesian MSE results in a Bayesian bound which is provably tight at both high and low SNR values. Numerical simulation of the location estimation problem demonstrates that the technique is both simpler and tighter than alternative approaches.

\section*{Acknowledgement}

The authors are grateful to Dr.\ Volker Pohl for fruitful discussions concerning many of the mathematical aspects of the paper. The authors would also like to thank the anonymous reviewers for their many constructive comments.

\appendices

\section{Some Technical Lemmas}
\label{ap:technical lemmas}

The proof of several theorems in the paper relies on the following technical results.

\begin{lemma} \label{le:uniq min gen}
Consider the minimization problems
\beq \label{eq:M ell}
M_\ell = \inf_{\b \in S} Z_\ell[\b], \quad \ell = 1,2,3
\eeq
where $\J(\bth)$ is positive definite and bounded a.e.\ ($p_\bth$),
\begin{align} \label{eq:def Z1 Z2}
Z_1[\b] &\triangleq \int_\Theta \|\b(\bth)\|^2 p_\bth(d\bth) \notag\\
Z_2[\b] &\triangleq \int_\Theta
    \Tr\!\left(
        \left(\I+\pd{\b}{\bth}\right) \J^{-1}(\bth) \left(\I+\pd{\b}{\bth}\right)^T
    \right) p_\bth(d\bth) \notag\\
Z_3[\b] &\triangleq Z_1[\b] + Z_2[\b]
\end{align}
and $S \subset \HoTn$ is convex, closed, and bounded under the $\HoTn$ norm \eqref{eq:HoTn inner prod}. Then, for each $\ell$, there exists a function $\b^{(0)} \in S$ such that $Z[\b^{(0)}]=M_\ell$. If $\ell = 1$ or $\ell = 3$, then the minimizer of \eqref{eq:M ell} is unique.
\end{lemma}

Note that $Z_3[\b]$ equals $Z[\b]$ of \eqref{eq:average CRB}; the notation $Z_3[\b]$ is introduced for simplicity. Also note that under mild regularity assumptions on $\J(\bth)$, uniqueness can be demonstrated for $\ell=2$ as well, but this is not necessary for our purposes.

\begin{proof}
The space $\HoTn$ is a Cartesian product of $n$ Sobolev spaces $H^1(\Theta)$, each of which is a separable Hilbert space \cite[\S3.7.1]{lebedev03}. Therefore, $\HoTn$ is also a separable Hilbert space. It follows from the Banach--Alaoglu theorem \cite[\S3.17]{rudin-FunctionalAnalysis} that all bounded sequences in $\HoTn$ have weakly convergent subsequences \cite[\S2.18]{lieb01}. Recall that a sequence $\f^{(1)}, \f^{(2)}, \ldots \in \HoTn$ is said to converge weakly to $\f^{(0)} \in \HoTn$ (denoted $\f^{(i)} \rightharpoonup \f^{(0)}$) if
\beq \label{eq:weak convergence}
L[\f^{(j)}] \rightarrow L[\f^{(0)}]
\eeq
for all continuous linear functionals $L[\cdot]$ \cite[\S2.9]{lieb01}.

Given a particular value $\ell \in \{1,2,3\}$, let $\b^{(i)}$ be a sequence of functions in $S$ such that $Z_\ell[\b^{(i)}] \rightarrow M_\ell$. This is a bounded sequence since $S$ is bounded, and therefore there exists a subsequence $\b^{(i_k)}$ which converges weakly to some $\bol \in \HoTn$. Furthermore, since $S$ is closed,\footnote{In fact, we require that $S$ be ``weakly closed'' in the sense that weakly convergent sequences in $S$ converge to an element in $S$. However, since $S$ is convex, this notion is equivalent to the ordinary definition of closure \cite[\S3.13]{rudin-FunctionalAnalysis}.} we have $\bol \in S$. We will now show that $Z_\ell[\bol] = M_\ell$.

To this end, it suffices to show that $Z_\ell[\cdot]$ is weakly lower semicontinuous, i.e., for any sequence $\f^{(i)} \in \HoTn$ which converges weakly to $\f^{(0)} \in \HoTn$, we must show that
\beq \label{eq:def weak lsc}
Z_\ell[\f^{(0)}] \le \liminf_{i\rightarrow\infty} Z_\ell[\f^{(i)}].
\eeq
Consider a weakly convergent sequence $\f^{(j)} \rightharpoonup \f^{(0)}$. Then, \eqref{eq:weak convergence} holds for any continuous linear functional $L[\cdot]$. Specifically, choose the continuous linear functional
\beq
L_1[\f] = \int_\Theta \f^{(0)}(\bth) \f(\bth) p_\bth(d\bth).
\eeq
We then have
\begin{align}
&Z_1[\f^{(0)}]
= L_1[\f^{(0)}] \notag\\
&= \lim_{j \rightarrow \infty} L_1[\f^{(j)}] \notag\\
&= \lim_{j \rightarrow \infty} \int_\Theta \sum_{i=1}^n f^{(0)}_i(\bth) f^{(j)}_i(\bth) p_\bth(d\bth) \notag\\
&\le \liminf_{j \rightarrow \infty} \sqrt{\int_\Theta \|\f^{(0)}(\bth)\|^2 p_\bth(d\bth) \cdot
     \int_\Theta \|\f^{(j)}(\bth)\|^2 p_\bth(d\bth)} \notag\\
&= \sqrt{Z_1[\f^{(0)}]} \liminf_{j \rightarrow \infty} \sqrt{Z_1[\f^{(j)}]}
\end{align}
where we have used the Cauchy--Schwarz inequality. It follows that
\beq
\sqrt{Z_1[\f^{(0)}]}
\le
\liminf_{j \rightarrow \infty}
\sqrt{Z_1[\f^{(j)}]}
\eeq
and therefore $Z_1[\f^{(0)}] \le \liminf_{j \rightarrow \infty} Z_1[\f^{(j)}]$, so that $Z_1[\cdot]$ is weakly lower semicontinuous.

Similarly, consider the continuous linear functional
\beq
L_2[\f] = \int_\Theta \Tr\!\left( \!
        \left(\I+\pd{\f^{(0)}}{\bth}\right) \! \J^{-1}(\bth) \! \left(\I+\pd{\f}{\bth}\right)^T
    \right) \! p_\bth(d\bth)
\eeq
for which we have
\begin{align} \label{eq:Z2prf}
&Z_2[\f^{(0)}] = L_2[\f^{(0)}] \notag\\
&= \lim_{j \rightarrow \infty} L_2[\f^{(j)}] \notag\\
&= \lim_{j \rightarrow \infty} \int_\Theta \Tr\BIG{22pt}[ \!
        \left(\I+\pd{\f^{(0)}}{\bth}\right) \J^{-1}(\bth) \notag\\
&\hspace{10em}       \cdot \left(\I+\pd{\f^{(j)}}{\bth}\right)^T \BIG{22pt}] p_\bth(d\bth).
\end{align}
Note that, for any positive definite matrix $\W$, $\Tr(\A \W {\boldsymbol B}^T)$ is an inner product of the two matrices $\A$ and ${\boldsymbol B}$. Therefore, by the Cauchy--Schwarz inequality,
\beq
\Tr(\A \W {\boldsymbol B}^T) \le \sqrt{ \Tr(\A \W \A^T) \Tr({\boldsymbol B} \W {\boldsymbol B}^T) }.
\eeq
Applying this to \eqref{eq:Z2prf}, we have
\begin{align}
Z_2&[\f^{(0)}] \le \liminf_{j \rightarrow \infty}
 \notag\\
\int_\Theta &\sqrt{\Tr\!\left( \left( \I+\pd{\f^{(0)}}{\bth} \right) \J^{-1}(\bth) \left( \I+\pd{\f^{(0)}}{\bth} \right)^T \right) } \notag\\
\cdot &
\sqrt{\Tr\!\left( \left( \I+\pd{\f^{(j)}}{\bth} \right) \J^{-1}(\bth) \left( \I+\pd{\f^{(j)}}{\bth} \right)^T \right) } p_\bth(d\bth).
\end{align}
Once again using the Cauchy--Schwarz inequality results in
\beq
Z_2[\f^{(0)}] \le \liminf_{j \rightarrow \infty} \sqrt{ Z_2[\f^{(0)}] Z_2[\f^{(j)}] }
\eeq
and therefore $Z_2[\f^{(0)}] \le \liminf_{j \rightarrow \infty} Z_2[\f^{(j)}]$, so that $Z_2[\cdot]$ is weakly lower semicontinuous. Since $Z_3[\f] = Z_1[\f] + Z_2[\f]$, it follows that $Z_3[\cdot]$ is also weakly lower semicontinuous.

Now recall that $\b^{(i_k)} \rightharpoonup \bol$ and $Z_\ell[\b^{(i_k)}] \rightarrow M_\ell$. By the definition \eqref{eq:def weak lsc} of lower semicontinuity, it follows that
\beq
Z_\ell[\bol] \le \liminf_{k \rightarrow \infty} Z_\ell[\b^{(i_k)}] = M_\ell
\eeq
and since $M_\ell$ is the infimum of $Z_\ell[\b]$, we obtain $Z[\bol] = M$. Thus $\bol$ is a minimizer of \eqref{eq:M ell}.

It remains to show that for $\ell \in \{1,3\}$, the minimizer of \eqref{eq:M ell} is unique.
To this end, we first show that $Z_1[\cdot]$ is strictly convex. Let $b^{(0)},b^{(1)} \in S$ be two essentially different functions, i.e.,
\beq
p_\bth\! \left( \left\{
    \bth \in \Theta : \b^{(0)}(\bth) \neq \b^{(1)}(\bth) \right\} \right) > 0.
\eeq
Let $\b^{(2)}(\bth) = \lambda \b^{(0)}(\bth) + (1-\lambda) \b^{(1)}(\bth)$ for some $0 < \lambda < 1$, so that $\b^{(2)} \in S$ by convexity. We then have
\begin{align}
Z_1[\b^{(2)}]
&= \int_Q
    \left\| \lambda \b^{(0)}(\bth) + (1-\lambda) \b^{(1)}(\bth) \right\|^2 p_\bth(d\bth) \notag\\
&+ \int_{\Theta \backslash Q}
    \left\| \lambda \b^{(0)}(\bth) + (1-\lambda) \b^{(1)}(\bth) \right\|^2 p_\bth(d\bth) \notag\\
&< \int_Q
    \left[ \lambda \|\b^{(0)}(\bth)\|^2 + (1-\lambda) \|\b^{(1)}(\bth)\|^2 \right] p_\bth(\bth) \notag\\
&+ \int_{\Theta \backslash Q}
    \left[ \lambda \|\b^{(0)}(\bth)\|^2 + (1-\lambda) \|\b^{(1)}(\bth)\|^2 \right] p_\bth(\bth) \notag\\
&= \lambda Z_1[\b^{(0)}] + (1-\lambda) Z_2[\b^{(1)}]
\end{align}
where the inequality follows from strict convexity of the squared Euclidean norm $\|\x\|^2$. Thus $Z_1[\cdot]$ is strictly convex, and hence has a unique minimum.

Note that $Z_3[\b] = Z_1[\b] + Z_2[\b]$. Since $Z_1[\cdot]$ is strictly convex and $Z_2[\cdot]$ is convex, it follows that $Z_3[\cdot]$ is strictly convex, and thus also has a unique minimum. This completes the proof.
\end{proof}

The following lemma can be thought of as a triangle inequality for a normed space of matrix functions over $\Theta$.

\begin{lemma} \label{le:triangle}
Let $p_\bth$ be a probability measure over $\Theta$, and let $\M : \Theta \rightarrow {\mathbb R}^{n \times n}$ be a matrix function. Suppose
\beq \label{eq:lem M given}
\int_\Theta \| \I + \M(\bth) \|_F^2 p_\bth(d\bth) \le \alpha
\eeq
for some constant $\alpha$. It follows that
\beq
\int_\Theta \|\M(\bth)\|_F^2 p_\bth(d\bth) \le (\sqrt{\alpha} + \sqrt{n})^2.
\eeq
\end{lemma}

\begin{proof}
By the triangle inequality,
\begin{align}
\left\|\M(\bth)\right\|_F
=   \left\|\M(\bth) + \I - \I\right\|_F
\le \left\|\M(\bth) + \I\right\|_F + \|\I\|_F .
\end{align}
Since $\|\I\|_F^2 = n$, we have
\begin{align}
&\int_\Theta \left\| \M(\bth) \right\|_F^2 p_\bth(d\bth)  \notag\\
&\le \int_\Theta \left[ \left\| \I + \M(\bth) \right\|_F^2 + n +
2\sqrt{n}\left\| \I + \M(\bth) \right\|_F \right] p_\bth(d\bth).
\end{align}
Using the fact that
\beq
\int_\Theta \left\| \I + \M(\bth) \right\|_F p_\bth(d\bth)
\le
\sqrt{ \int_\Theta \left\| \I + \M(\bth) \right\|_F^2 p_\bth(d\bth) }
\eeq
and combining with \eqref{eq:lem M given}, it follows that
\begin{align}
\int_\Theta \left\| \M(\bth) \right\|_F^2 p_\bth(d\bth)
&\le
\alpha + n + 2\sqrt{n \alpha}
\end{align}
which completes the proof.
\end{proof}

\section{Proof of Proposition~\ref{pr:uniq min}}
\label{ap:prf pr:uniq min}

The following proof of Proposition~\ref{pr:uniq min} makes use of the results developed in Appendix~\ref{ap:technical lemmas}.

\begin{proof}[Proof of Proposition~\ref{pr:uniq min}]
Recall that $Z_3[\b]$ of \eqref{eq:def Z1 Z2} equals $Z[\b]$. Thus, we would like to apply Lemma~\ref{le:uniq min gen} (with $\ell=3$) to prove the unique existence of a minimizer of \eqref{eq:min HoTn}. However, Lemma~\ref{le:uniq min gen} requires that the minimization be performed over a closed, bounded, and convex set $S$, whereas \eqref{eq:min HoTn} is performed over the unbounded set $\HoTn$. To resolve this issue, we must show that the minimization \eqref{eq:min HoTn} can be reformulated as a minimization over a closed, bounded, and convex set $S$.

To this end, note that
\beq
Z[{\bf 0}] = \int_\Theta \Tr(\J^{-1}(\bth)) p_\bth(d\bth) \triangleq U
\eeq
and therefore $M \le U < \infty$. Thus, it suffices to perform the minimization \eqref{eq:min HoTn} over those functions for which $Z[\b] \le U$. We now show that this can be achieved by minimizing over a closed, bounded, and convex set $S$. First, note that $Z[\b] \ge \|\b\|_\LtTn^2$, so that one may choose to minimize \eqref{eq:min HoTn} only over functions $\b$ for which
\beq \label{eq:limit LtTn norm}
\|\b\|_\LtTn^2 \le U.
\eeq
Similarly, we have
\beq
Z[\b] \ge
\int_\Theta
\Tr\!\left( \left(\I+\pd{\b}{\bth}\right) \J^{-1}(\bth) \left(\I+\pd{\b}{\bth}\right)^T \right)
p_\bth(d\bth)
\eeq
so that it suffices to minimize \eqref{eq:min HoTn} over functions $\b$ for which
\beq \label{eq:lem1prf1}
\int_\Theta
\Tr\!\left( \left(\I+\pd{\b}{\bth}\right) \J^{-1}(\bth) \left(\I+\pd{\b}{\bth}\right)^T \right)
p_\bth(d\bth)
\le U.
\eeq
Note that $\J(\bth)$ is bounded a.e., and therefore $\lambda_{\min}(\J^{-1}) \ge 1/K$ a.e., for some constant $K$. It follows that
\begin{multline}
\Tr\!\left( \left(\I+\pd{\b}{\bth}\right) \J^{-1}(\bth) \left(\I+\pd{\b}{\bth}\right)^T \right)\\
\ge \frac{1}{K} \left\| \I+\pd{\b}{\bth} \right\|_F^2\ \text{ a.e.}(p_\bth).
\end{multline}
Combining with \eqref{eq:lem1prf1} yields
\beq \label{eq:lem1prf2}
\int_\Theta \left\| \I + \pd{\b}{\bth} \right\|_F^2 p_\bth(d\bth) \le KU.
\eeq
From Lemma~\ref{le:triangle}, we then have
\begin{align} \label{eq:limit deriv norm}
\int_\Theta \left\| \pd{\b}{\bth} \right\|_F^2 p_\bth(d\bth)
&\le
\left( \sqrt{n} + \sqrt{KU} \right)^2.
\end{align}
From \eqref{eq:limit LtTn norm} and \eqref{eq:limit deriv norm} it follows that the minimization \eqref{eq:min HoTn} can be limited to the closed, bounded, convex set
\beq
S = \left\{ \b \in \HoTn : \|\b\|_\HoTn^2 \le U + \left( \sqrt{KU} + \sqrt{n} \right)^2 \right\}.
\eeq
Applying Lemma~\ref{le:uniq min gen} proves the unique existence of a minimizer of \eqref{eq:min HoTn}. The proof that $0 < s < \infty$ appears immediately after the statement of Proposition~\ref{pr:uniq min}.
\end{proof}

\section{Proof of Theorem~\ref{TH:E-L}}
\label{ap:prf th:e-l}

The following is the proof of Theorem~\ref{TH:E-L} concerning the calculation of the OBB\@.

\begin{proof}[Proof of Theorem~\ref{TH:E-L}]
Consider the more general problem of minimizing the functional
\beq
Z[\b] = \int_\Theta F[\b,\bth] d\bth
\eeq
where $F[\b,\bth]$ is smooth and convex in $\b: \Theta \rightarrow {\mathbb R}^n$, and $\Theta \subset {\mathbb R}^n$ is a bounded set with a smooth boundary $\Lambda$. Then, $Z[\b]$ is also smooth and convex in $\b$, so that $\b$ is a global minimum of $Z[\b]$ if and only if the differential $\delta Z[\h]$ equals zero at $\b$ for all admissible functions $\h: \Theta \rightarrow {\mathbb R}^n$ \cite{gelfand00}.

By a standard technique \cite[\S 35]{gelfand00}, it can be shown that
\begin{align} \label{eq:delta C}
&\delta Z[\h] =
\epsilon \sum_i \int_\Theta \left( \pd{F}{b_i} - \sum_j \pd{}{\theta_j} \pd{F}{b_i^{(j)}} \right) h_i(\bth) d\bth \notag\\
&\quad +
\epsilon \sum_i \int_\Lambda \left( \pd{F}{b_i^{(1)}}, \ldots, \pd{F}{b_i^{(n)}} \right)^T \bnu(\bth) \, h_i(\bth) \, d\sigma
\end{align}
where $\epsilon$ is an infinitesimal quantity, $b_i^{(j)} = \partial b_i / \partial \theta_j$, and $\bnu(\bth)$ is an outward-pointing normal at the boundary point $\bth \in \Lambda$.
We now seek conditions for which $\delta Z[\h] = 0$ for all $\h(\bth)$. Consider first functions $\h(\bth)$ which equal zero on the boundary $\Lambda$. In this case, the second integral vanishes, and we obtain the Euler--Lagrange equations
\beq \label{eq:E-L general}
\forall i, \ \
\pd{F}{b_i} - \sum_j \pd{}{\theta_j} \pd{F}{b_i^{(j)}} = 0.
\eeq
Substituting this result back into \eqref{eq:delta C}, and again using the fact that $\delta Z[\h]=0$ for all $\h$, we obtain the boundary condition
\beq \label{eq:boundary general}
\forall i,\ \forall \bth \in \Lambda, \ \
\left( \pd{F}{b_i^{(1)}}, \ldots, \pd{F}{b_i^{(n)}} \right)^T \bnu(\bth) = 0.
\eeq
Plugging $F[\b,\bth] = \CRB[\b,\bth] p_\bth(\bth)$ into \eqref{eq:E-L general} and \eqref{eq:boundary general} provides the required result.
\end{proof}

\section{Proof of Theorem~\ref{TH:SPH SYM}}
\label{ap:prf th:sph sym}

Before proving Theorem~\ref{TH:SPH SYM}, we provide the following two lemmas, which demonstrate some symmetry properties of the CRB\@.

\begin{lemma} \label{le:invariant}
Under the conditions of Theorem~\ref{TH:SPH SYM}, the functional $Z[\b]$ of \eqref{eq:average CRB} is rotation and reflection invariant, i.e., $Z[\b] = Z[\U \b]$ for any unitary matrix $\U$.
\end{lemma}

\begin{proof}
We first demonstrate that $Z[\b]$ is rotation invariant. From the definitions of $Z[\b]$ and $\CRB[\b,\bth]$, we have
\begin{align} \label{eq:le:1}
Z[\b]
&= \int_\Theta \Tr\!\left[\left(\I+\pd{\b}{\bth}\right) \left(\I+\pd{\b}{\bth}\right)^T \right] \frac{q(\|\bth\|)}{J(\|\bth\|)} d\bth \notag\\
&+ \int_\Theta \|\b(\bth)\|^2 q(\|\bth\|) d\bth.
\end{align}
The second integral is clearly rotation invariant, since a rotation of $\b$ does not alter its norm. It remains to show that the first integral, which we denote by $I_1[\b]$, does not change when $\b$ is rotated. To this end, we begin by considering a rotation about the first two coordinates, such that $\b$ is transformed to $\tilde{\b} \triangleq \R_\phi \b$, where the rotation matrix $\R_\phi$ is defined such that
\begin{align}
\R_\phi \b =
(&b_1 \cos\phi + b_2 \sin\phi, \notag\\
&-b_1 \sin\phi + b_2 \cos\phi,
b_3, \ldots, b_n)^T.
\end{align}
We must thus show that $I_1[\b] = I_1[\tilde{\b}]$. Let us perform the change of variables $\bth \mapsto \tilde{\bth}$, where $\tilde{\bth} = \R_{(-\phi)} \bth$. Rewriting the trace in \eqref{eq:le:1} as a sum, we have
\beq \label{eq:le:2}
I_1[\tilde{\b}] = \int_\Theta
\sum_{i,j} \left( \delta_{ij} + \pd{\tilde{b}_i}{\theta_j} \right)^2
\frac{q(\|\tilde{\bth}\|)}{J(\|\tilde{\bth}\|)} d\tilde{\bth}
\eeq
where we have used the facts that $\|\bth\|=\|\tilde{\bth}\|$ and that $\Theta$ does not change under the change of variables.

We now demonstrate some properties of the transformation of $\b$ and $\bth$. First, we have, for any $j$,
\begin{align} \label{eq:le:inv1}
   \left( \pd{\tilde{b}_1}{\theta_j} \right)^2
 + \left( \pd{\tilde{b}_2}{\theta_j} \right)^2
&= \left( \pd{b_1}{\theta_j} \cos\phi + \pd{b_2}{\theta_j} \sin\phi \right)^2 \notag\\
&\quad+ \left(-\pd{b_1}{\theta_j} \sin\phi + \pd{b_2}{\theta_j} \cos\phi \right)^2 \notag\\
&= \left( \pd{b_1}{\theta_j} \right)^2 + \left( \pd{b_2}{\theta_j} \right)^2.
\end{align}
Also, for any $i$,
\begin{align} \label{eq:le:inv2}
   \left( \pd{b_i}{\tilde{\theta}_1} \right)^2
 + \left( \pd{b_i}{\tilde{\theta}_2} \right)^2
&= \left( \pd{b_i}{\theta_1} \pd{\theta_1}{\tilde{\theta}_1}
        + \pd{b_i}{\theta_2} \pd{\theta_2}{\tilde{\theta}_1} \right)^2 \notag\\
&\quad + \left( \pd{b_i}{\theta_1} \pd{\theta_1}{\tilde{\theta}_2}
        + \pd{b_i}{\theta_2} \pd{\theta_2}{\tilde{\theta}_2} \right)^2 \notag\\
&= \left( \pd{b_i}{\theta_1} \right)^2
 + \left( \pd{b_i}{\theta_2} \right)^2
\end{align}
where we used the fact that $\bth = \R_{\phi}\tilde{\bth}$. Third, we have
\begin{align}
\pd{\tilde{b}_1}{\theta_1} &= \pd{b_1}{\tilde{\theta}_1} \cos^2\phi
                            + \pd{b_1}{\tilde{\theta}_2} \sin\phi \cos\phi \notag\\
                           &\quad+ \pd{b_2}{\tilde{\theta}_1} \sin\phi \cos\phi
                            + \pd{b_2}{\tilde{\theta}_2} \sin^2\phi, \notag\\
\pd{\tilde{b}_2}{\theta_2} &= \pd{b_1}{\tilde{\theta}_1} \sin^2\phi
                            - \pd{b_1}{\tilde{\theta}_2} \sin\phi \cos\phi \notag\\
                           &\quad- \pd{b_2}{\tilde{\theta}_1} \sin\phi \cos\phi
                            + \pd{b_2}{\tilde{\theta}_2} \cos^2\phi,
\end{align}
so that
\beq \label{eq:le:inv3}
\pd{\tilde{b}_1}{\theta_1} + \pd{\tilde{b}_2}{\theta_2}
 = \pd{b_1}{\tilde{\theta}_1} + \pd{b_2}{\tilde{\theta}_2}.
\eeq
We now show that
\beq \label{eq:le:sum invariant}
  \sum_{i,j} \left( \delta_{ij} + \pd{\tilde{b}_i}{\theta_j} \right)^2
= \sum_{i,j} \left( \delta_{ij} + \pd{b_i}{\tilde{\theta}_j} \right)^2.
\eeq
For terms with $i,j \ge 3$, we have $b_i = \tilde{b}_i$ and $\theta_j = \tilde{\theta}_j$, so that replacing $\tilde{\b}$ with $\b$ and $\bth$ with $\tilde{\bth}$ does not change the result. The terms with $i=1,2$ and $j \ge 3$ do not change because of \eqref{eq:le:inv1}, while the terms with $i \ge 3$ and $j = 1,2$ do not change because of \eqref{eq:le:inv2}. It remains to show that the terms $i,j=1,2$ do not modify the sum. To this end, we write out these four terms as
\begin{align}
 & \left( 1 + \pd{\tilde{b}_1}{\theta_1} \right)^2
 + \left( 1 + \pd{\tilde{b}_2}{\theta_2} \right)^2
 + \left( \pd{\tilde{b}_1}{\theta_2} \right)^2
 + \left( \pd{\tilde{b}_2}{\theta_1} \right)^2 \notag\\
&\quad
 = 2 + 2\pd{\tilde{b}_1}{\theta_1} + 2\pd{\tilde{b}_2}{\theta_2} \notag\\
&\qquad + \left( \pd{\tilde{b}_1}{\theta_1} \right)^2
 + \left( \pd{\tilde{b}_1}{\theta_2} \right)^2
 + \left( \pd{\tilde{b}_2}{\theta_1} \right)^2
 + \left( \pd{\tilde{b}_2}{\theta_2} \right)^2 \notag\\
&\quad
 = 2 + 2\pd{{b}_1}{\tilde\theta_1} + 2\pd{{b}_2}{\tilde\theta_2} \notag\\
&\qquad + \left( \pd{{b}_1}{\tilde\theta_1} \right)^2
 + \left( \pd{{b}_1}{\tilde\theta_2} \right)^2
 + \left( \pd{{b}_2}{\tilde\theta_1} \right)^2
 + \left( \pd{{b}_2}{\tilde\theta_2} \right)^2 \notag\\
&\quad
 = \left( 1 + \pd{{b}_1}{\tilde\theta_1} \right)^2
 + \left( 1 + \pd{{b}_2}{\tilde\theta_2} \right)^2
 + \left( \pd{{b}_1}{\tilde\theta_2} \right)^2
 + \left( \pd{{b}_2}{\tilde\theta_1} \right)^2
\end{align}
where, in the second transition, we have used \eqref{eq:le:inv1}, \eqref{eq:le:inv2}, and \eqref{eq:le:inv3}. It follows that $I_1[\tilde{\b}]$ of \eqref{eq:le:2} is equal to $I_1[\b]$, and hence $Z[\b] = Z[\tilde{\b}]$. The result similarly holds for rotations about any other two coordinates. Since any rotation can be decomposed into a sequence of two-coordinate rotations, we conclude that $Z[\b]$ is rotation invariant.

Next, we prove that $Z[\b]$ is invariant to reflections through hyperplanes containing the origin. Since $Z[\b]$ is invariant to rotations, it suffices to choose a single hyperplane, say $\{\bth: \theta_1 = 0\}$. Let
\beq
\tilde{\b} \triangleq (-b_1(\bth), b_2(\bth), \ldots, b_n(\bth))^T
\eeq
be the reflection of $\b$, and consider the corresponding change of variables
\beq
\tilde{\bth} \triangleq (-\theta_1, \theta_2, \ldots, \theta_n)^T.
\eeq
By the symmetry assumptions, $p_\bth$ and $\J$ are unaffected by the change of variables; furthermore, $\partial \tilde{\b} / \partial \tilde{\bth} = \partial \b / \partial \bth$. It follows that $\CRB[\tilde{\b},\tilde{\bth}] = \CRB[\b,\bth]$, and therefore $Z[\b] = Z[\tilde{\b}]$.
\end{proof}

\begin{lemma} \label{le:invariant2}
Suppose $\b(\bth)$ is radial and rotation invariant, i.e., $\b(\bth) = t(\|\bth\|^2) \bth$ for some function $t \in \HoTn$. Also suppose that $\J(\bth) = J(\|\bth\|) \I$, where $J(\cdot)$ is a scalar function. Then, $\CRB[\b,\bth]$ of \eqref{eq:CRB} is rotation invariant in $\bth$, i.e., $\CRB[\b,\R\bth] = \CRB[\b,\bth]$ for any rotation matrix $\R$.
\end{lemma}

\begin{proof}
We will show that $\CRB[\b,\bth]$ depends on $\bth$ only through $\|\bth\|^2$, and is therefore rotation invariant. For the given value of $\b(\bth)$ and $\J(\bth)$, we have
\begin{align}
&\CRB[\b,\bth] \notag\\
&= \|\b(\bth)\|^2 + \Tr\left[ \left( \I + \pd{\b}{\bth} \right) \J^{-1}(\bth)
                              \left( \I + \pd{\b}{\bth} \right)^T \right] \notag\\
&= t^2 \|\bth\|^2 + \frac{1}{J(\|\bth\|)}
                    \Tr\left[ \left( \I + \pd{t\bth}{\bth} \right)
                              \left( \I + \pd{t\bth}{\bth} \right)^T \right]
\end{align}
where, for notational convenience, we have omitted the dependence of $t$ on $\|\bth\|^2$. It remains to show that the trace in the above expression is a function of $\bth$ only through $\|\bth\|^2$. To this end, we note that
\begin{align}
\pd{b_i}{\theta_j}
&= t \delta_{ij} + t' \theta_i \pd{\|\bth\|^2}{\theta_j}
 = t \delta_{ij} + 2 t' \theta_i \theta_j
\end{align}
where $\delta_{ij}$ is the Kronecker delta. It follows that
\beq
\left( \delta_{ij} + \pd{b_i}{\theta_j} \right)^2
= (1+t)^2 \delta_{ij} + 4(1+t) t' \theta_i \theta_j \delta_{ij}
  + 4 t'^2 \theta_i^2 \theta_j^2.
\eeq
Therefore
\begin{align}
\Tr&\left[ \left( \I + \pd{\b}{\bth} \right)
                              \left( \I + \pd{\b}{\bth} \right)^T \right]
= \sum_{i,j} \left( \delta_{ij} + \pd{b_i}{\theta_j} \right)^2 \notag\\
&= n(1+t)^2 + 4t'^2 \sum_{i,j} \theta_i^2 \theta_j^2 + 4(1+t)t' \sum_i \theta_i^2 \notag\\
&= n(1+t)^2 + 4t'^2 \|\bth\|^4 + 4(1+t)t' \|\bth\|^2.
\end{align}
Thus, $\CRB[\b,\bth]$ depends on $\bth$ only through $\|\bth\|^2$, completing the proof.
\end{proof}

\begin{proof}[Proof of Theorem~\ref{TH:SPH SYM}]
We have seen in Theorem~\ref{TH:E-L} that the solution of \eqref{eq:th:biased bayesian CRB} is unique. Now suppose that the optimum $\b$ is not rotation invariant, i.e., there exists a rotation matrix $\R$ such that $\R \b(\bth)$ is not identical to $\b(\bth)$. By Lemma~\ref{le:invariant}, $\R \b(\bth)$ is also optimal, which is a contradiction.

Furthermore, suppose that $\b$ is not radial, i.e., for some value of $\bth$, $\b(\bth)$ contains a component perpendicular to the vector $\bth$. Consider a hyperplane passing through the origin, whose normal is the aforementioned perpendicular component. By Lemma~\ref{le:invariant}, The reflection of $\b$ through this hyperplane is also an optimal solution of \eqref{eq:th:biased bayesian CRB}, which is again a contradiction. Therefore, the optimum $\b$ is spherically symmetric and radial, so that it can be written as
\beq \label{eq:th3 opt bias}
\b(\bth) =
b(\|\bth\|) \frac{\bth}{\|\bth\|}   
\eeq
where $b(\cdot)$ is a scalar function.

To determine the value of $b(\cdot)$, it suffices to analyze the differential equation \eqref{eq:th:E-L} along a straight line from the origin to the boundary. We choose a line along the $\theta_1$ axis, and begin by calculating the derivatives of $b_1(\bth)$, $q(\|\bth\|)$, and $\J(\|\bth\|)$ along this axis. The derivative of $q(\|\bth\|)$ is given by
\beq
\pd{q}{\theta_j} = q'(\rho) \frac{\theta_j}{\rho}
\eeq
where we have denoted $\rho = \|\bth\|$, so that $\rho$ is weakly differentiable and
\beq
\pd{\rho}{\theta_j} = \frac{\theta_j}{\rho}.
\eeq
Along the $\theta_1$ axis, we have $\theta_1 = \rho$ while $\theta_2 = \cdots = \theta_n = 0$, so that
\beq
\left. \pd{q}{\theta_j} \right|_{\bth = \rho\e_1} = q'(\rho) \delta_{j1}.
\eeq
Similarly, since $\J(\bth) = J(\rho) \I$,
\beq
\pd{(\J^{-1})_{jk}}{\theta_j} = - \frac{J'(\rho)}{J^2(\rho)} \frac{\theta_j}{\rho} \delta_{jk}
\eeq
so that along the $\theta_1$ axis
\beq
\left. \pd{(\J^{-1})_{jk}}{\theta_j} \right|_{\bth = \rho\e_1} = - \frac{J'(\rho)}{J^2(\rho)} \delta_{jk} \delta_{j1}.
\eeq
From \eqref{eq:th3 opt bias}, we have
\beq \label{eq:b first deriv}
\pd{b_i}{\theta_j}
 = b'(\rho) \frac{\theta_i \theta_j}{\rho^2}
 + \frac{b(\rho)}{\rho} \left( \delta_{ij} - \frac{\theta_i \theta_j}{\rho^2} \right).
\eeq
Thus, on the $\theta_1$ axis, we have
\begin{align}
\left. \pd{b_1}{\theta_j} \right|_{\bth = \rho\e_1} &= b'(\rho) \delta_{j1}.
\end{align}
The second derivative of $b_i(\bth)$ can be shown to equal
\begin{align}
&\pdd{b_i}{\theta_j}{\theta_k}
 = b''(\rho) \frac{\theta_i \theta_j \theta_k}{\rho^3}\notag\\
&+ \left( \frac{b'(\rho)}{\rho} - \frac{b(\rho)}{\rho^2} \right)
   \left( \frac{\theta_i}{\rho} \delta_{jk}
        + \frac{\theta_j}{\rho} \delta_{ik}
        + \frac{\theta_k}{\rho} \delta_{ij}
        - 3 \frac{\theta_i \theta_j \theta_k}{\rho^3} \right).
\end{align}
Therefore, on the $\theta_1$ axis
\begin{align}
\left. \frac{\partial^2 b_1}{\partial \theta_1^2} \right|_{\bth = \rho\e_1}
&= b''(\rho) \notag\\
\left. \frac{\partial^2 b_1}{\partial \theta_j^2} \right|_{\bth = \rho\e_1}
&= \frac{b'(\rho)}{\rho}-\frac{b(\rho)}{\rho^2} & (j &\ne 1) \notag\\
\left. \pdd{b_1}{\theta_j}{\theta_k} \right|_{\bth = \rho\e_1}
&= 0 & (j,k &\ne 1).
\end{align}
Substituting these derivatives into \eqref{eq:th:E-L}, we obtain
\begin{align}
q(\rho) b(\rho)
&= \frac{q(\rho)}{J(\rho)} \left(b''(\rho) + (n-1)\frac{b'(\rho)}{\rho}
                                   - (n-1)\frac{b(\rho)}{\rho^2} \right) \notag\\
&+ (1 + b'(\rho)) \left( \frac{q'(\rho)}{J(\rho)}
                        - q(\rho)\frac{J'(\rho)}{J^2(\rho)} \right)
\end{align}
which is equivalent to \eqref{eq:th:sph sym ode}.

To obtain the boundary conditions, observe that Lemma~\ref{le:invariant} implies $\b({\bf 0}) = {\bf 0}$, whence we conclude that $b(0)=0$. Next, evaluate the boundary condition \eqref{eq:th:E-L:boundary} at boundary point $\bth = r \e_1$, where the surface normal $\bnu(\bth)$ equals $\e_1$, so that
\beq
1 + b'(\rho) = 1 + \pd{b_1}{\theta_1} = 0, \quad \bth = r \e_1
\eeq
which is equivalent to the boundary condition $b'(r)=-1$.

To find the OBB \eqref{eq:th:sph sym}, we must now calculate $Z[\b]$ for the obtained bias function \eqref{eq:th3 opt bias}. To this end, note that, by Lemma~\ref{le:invariant2}, $\CRB[\b,\bth]$ is rotation invariant in $\bth$ for the required $\b(\bth)$. Thus, the integrand $\CRB[\b,\bth] q(\|\bth\|)$ is constant on any $(n-1)$-sphere centered on the origin, so that
\beq \label{eq:th3 Zb}
Z[\b] = \int_0^r \CRB[\b, \rho \e_1] q(\rho) S_n(\rho) d\rho
\eeq
where
\beq
S_n(\rho) = \frac{2 \pi^{n/2}}{\Gamma(n/2)} \rho^{n-1}
\eeq
is the hypersurface area of an $(n-1)$-sphere of radius $\rho$ \cite{vinogradov95}. It thus suffices to calculate the value of $\CRB[\b,\bth]$ at points along the $\theta_1$ axis. From \eqref{eq:b first deriv}, it follows that
\beq
\left. \pd{\b}{\bth} \right|_{\bth = \rho \e_1}
 = \diag\left( b'(\rho), \frac{b(\rho)}{\rho}, \ldots, \frac{b(\rho)}{\rho} \right).
\eeq
Substituting this into the definition of $\CRB[\b,\bth]$, we obtain
\begin{align} \label{eq:th3 almost done}
&\CRB[\b,\rho \e_1] \notag\\
& = b^2(\rho)
 + \frac{1}{J(\rho)} (1 + b'(\rho))^2
 + \frac{n-1}{J(\rho)} \left( 1 + \frac{b(\rho)}{\rho} \right)^2.
\end{align}
Combining \eqref{eq:th3 almost done} with \eqref{eq:th3 Zb} yields \eqref{eq:th:sph sym}, as required.
\end{proof}

\section{Proofs of Asymptotic Properties}
\label{ap:asymp}

Theorems \ref{th:low SNR} and \ref{th:high SNR} demonstrate asymptotic tightness of the OBB\@. The proofs of these two theorems follow.

\begin{proof}[Proof of Theorem~\ref{th:low SNR}]
We begin the proof by studying a certain optimization problem, whose relevance will be demonstrated shortly. Let $t \ge 0$ be a constant and consider the problem
\begin{align} \label{eq:def u(t)}
u(t) = \inf_{\b \in \HoTn} &\int_\Theta \left\| \I + \pd{\b}{\bth} \right\|_F^2 p_\bth(d\bth) \notag\\
       \text{s.t. } & \int_\Theta \|\b(\bth)\|^2 p_\bth(d\bth) \le t.
\end{align}
Notice that $u(t) \le n$ for all $t$, since an objective having a value of $n$ is achieved by the function $\b(\bth) = {\bf 0}$. Thus, it suffices to perform the minimization \eqref{eq:def u(t)} over functions $\b \in \HoTn$ satisfying
\beq \label{eq:lsnr prf 1}
\int_\Theta \left\| \I + \pd{\b}{\bth} \right\|_F^2 p_\bth(d\bth) \le n.
\eeq
It follows from Lemma~\ref{le:triangle} that such functions also satisfy
\beq \label{eq:lsnr prf 2}
\int_\Theta \left\| \pd{\b}{\bth} \right\|_F^2 p_\bth(d\bth) \le (2\sqrt{n})^2 = 4n.
\eeq
Therefore, \eqref{eq:def u(t)} is equivalent to the minimization
\beq \label{eq:lsnr prf 22}
u(t) = \inf_{\b \in S_t} \int_\Theta \left\| \I + \pd{\b}{\bth} \right\|_F^2 p_\bth(d\bth)
\eeq
where
\begin{align}
S_t = \bigg\{ \b \in \HoTn : & \int_\Theta \|\b(\bth)\|^2 p_\bth(d\bth) \le t, \notag\\
                           & \int_\Theta \left\| \pd{\b}{\bth} \right\|_F^2 p_\bth(d\bth) \le 4n \bigg\}.
\end{align}
The set $S_t$ is convex, closed, and bounded in $\HoTn$. Applying Lemma~\ref{le:uniq min gen} (with $\ell=2$) implies that there exists a function $\bopt \in S_t$ which minimizes \eqref{eq:lsnr prf 22}, and hence also minimizes \eqref{eq:def u(t)}.

Note that the objective in \eqref{eq:def u(t)} is zero if and only if
\beq \label{eq:lsnr prf 23}
\pd{\bopt}{\bth} = -\I \quad \text{a.e. }(p_\bth).
\eeq
The only functions in $\HoTn$ satisfying this requirement are the functions
\beq \label{eq:lsnr prf 24}
\b(\bth) = \k - \bth \quad \text{a.e. } (p_\bth)
\eeq
for some constant $\k \in {\mathbb R}^n$. Let $\bmu \triangleq \E{\bth}$ and define
\beq \label{eq:def v}
v \triangleq \E{\|\bth - \E{\bth}\|^2}.
\eeq
For functions of the form \eqref{eq:lsnr prf 24}, the constraint of \eqref{eq:def u(t)} is given by
\begin{align} \label{eq:lsnr prf 25}
   \int_\Theta \|\k-\bth\|^2 p_\bth(d\bth)
&= \int_\Theta \|\k-\bmu + \bmu-\bth\|^2 p_\bth(d\bth)  \notag\\
&= \|\k-\bmu\|^2 + v                                    \notag\\
&\ge v.
\end{align}
In \eqref{eq:lsnr prf 25}, equality is obtained if and only if $\k = \bmu$. Therefore, if $t<v$, no functions satisfying \eqref{eq:lsnr prf 23} are feasible, and thus
\begin{align} \label{eq:lsnr prf 3}
u(t) &= 0 \quad \text{if } t \ge v, \notag\\
u(t) &> 0 \quad \text{if } t < v.
\end{align}

We now return to the setting of Theorem~\ref{th:low SNR}. We must show that $\beta_N \rightarrow v$ as $N \rightarrow \infty$. We denote functions corresponding to the problem of estimating $\bth$ from $\x^{(N)}$ with a superscript $(N)$. Thus, for example, $Z^{(N)}[\b]$ denotes the functional $Z[\b]$ of \eqref{eq:average CRB} for the problem corresponding to the measurement vector $\x^{(N)}$.

Since all eigenvalues of $\J^{(N)}(\bth)$ decrease monotonically with $N$ for $p_\bth$-almost all $\bth$, we have
\beq
\CRB^{(N)}[\b,\bth] \le \CRB^{(N+1)}[\b,\bth]
\eeq
for any $\b \in \HoTn$, for $p_\bth$-almost all $\bth$, and for all $N$. Therefore
\beq
Z^{(N)}[\b] \le Z^{(N+1)}[\b].
\eeq
for any $\b \in \HoTn$ and for all $N$. It follows that for all $N$
\beq
\beta_N = \min_{\b \in \HoTn} Z^{(N)}[\b] \le
          \min_{\b \in \HoTn} Z^{(N+1)}[\b] = \beta_{N+1}
\eeq
so that $\beta_N$ is a non-decreasing sequence. Furthermore, note that
\beq
Z^{(N)}[\bmu-\bth] = v \quad \text{for all }N
\eeq
where $v$ is given by \eqref{eq:def v}. Therefore, $\beta_N \le v$ for all $N$. Thus $\beta_N$ converges to some value $q$, and we have
\beq \label{eq:betaN bounded}
\beta_N \le q \le v \quad \text{for all }N.
\eeq
To prove the theorem, it remains to show that $q=v$.

Let $\b^{(N)}$ be the minimizer of \eqref{eq:min HoTn} when $\bth$ is estimated from $\x^{(N)}$; this minimizer exists by virtue of Proposition~\ref{pr:uniq min}. We then have
\beq
\beta_N = Z^{(N)}[\b^{(N)}] \le q
\eeq
and therefore
\beq
\int_\Theta \|\b^{(N)}(\bth)\|^2 p_\bth(d\bth) \le q.
\eeq
It follows that $\b^{(N)}$ satisfies the constraint of the optimization problem \eqref{eq:def u(t)} with $t=q$. As a consequence, we have
\beq
\int_\Theta \left\| \I + \pd{\b^{(N)}}{\bth} \right\|_F^2 p_\bth(d\bth) \ge u(q).
\eeq
Define
\beq
\lambda_N \triangleq \esssup_{\bth\in\Theta} \lambda_{\max}(\J^{(N)}(\bth))
\eeq
and note that $\lambda_N > 0$ for all $N$, since $\J^{(N)}(\bth)$ is positive definite. Thus
\begin{align}
Z^{(N)}[\b^{(N)}]
&\ge \int_\Theta \Tr \BIG{22pt}[ \! \left( \I + \pd{\b^{(N)}}{\bth} \right)
                       \left( \J^{(N)}(\bth) \right)^{-1} \notag\\
&\hspace{6em}          \cdot \left( \I + \pd{\b^{(N)}}{\bth} \right)^T \BIG{22pt}]
        p_\bth(d\bth) \notag\\
&\ge \frac{1}{\lambda_N}
     \int_\Theta \left\| \I + \pd{\b^{(N)}}{\bth} \right\|_F^2 p_\bth(d\bth) \notag\\
&\ge \frac{u(q)}{\lambda_N}.
\end{align}
Assume by contradiction that $q<v$. From \eqref{eq:lsnr prf 3}, it then follows that $u(q)>0$. Since all eigenvalues of $\J^{(N)}(\bth)$ decrease to zero, we have $\lambda_N \rightarrow 0$, and thus
\beq
\beta_N \ge \frac{u(q)}{\lambda_N} \rightarrow \infty.
\eeq
This contradicts the fact \eqref{eq:betaN bounded} that $\beta_N \le v$. We conclude that $q=v$, as required.
\end{proof}

\begin{proof}[Proof of Theorem~\ref{th:high SNR}]
The proof is analogous to that of Theorem~\ref{th:low SNR}. We begin by considering the optimization problem
\begin{align} \label{eq:prf high 1}
\inf_{\b \in \HoTn} &\int_\Theta \|\b(\bth)\|^2 p_\bth(d\bth) \notag\\
\text{s.t. } &\int_\Theta \Tr\!\left( \left( \I+\pd{\b}{\bth} \right) \J^{-1}(\bth)
                                      \left( \I+\pd{\b}{\bth} \right)^T \right)
               p_\bth(d\bth) \le t
\end{align}
for some constant $t \ge 0$. Denote the minimum value of \eqref{eq:prf high 1} by $w(t)$. Let $\bmu = \E{\bth}$ and note that $\b(\bth) = \bmu-\bth$ satisfies the constraint in \eqref{eq:prf high 1} for any $t \ge 0$, and has an objective equal to $v$ of \eqref{eq:def v}. Thus, to determine $w(t)$, it suffices to minimize \eqref{eq:prf high 1} over the set
\begin{multline}
S_t = \bigg\{ \b \in \HoTn : \int_\Theta \|\b(\bth)\|^2 p_\bth(d\bth) \le v, \notag\\
                   \int_\Theta \Tr\!\left( \left(\I+\pd{\b}{\bth}\right) \J^{-1}(\bth)
                                            \left(\I+\pd{\b}{\bth}\right)^T \right)
                        p_\bth(d\bth) \le t \bigg\}.
\end{multline}
Define
\beq
\lambda \triangleq \esssup_{\bth \in \Theta} \lambda_{\max}(\J(\bth)).
\eeq
Since $\J(\bth)$ is positive definite almost everywhere, we have $\lambda>0$. For any $\b \in S_t$, we have
\beq
\frac{1}{\lambda} \int_\Theta \left\| \I + \pd{\b}{\bth} \right\|_F^2 p_\bth(d\bth) \le t
\eeq
and therefore, by Lemma~\ref{le:triangle},
\beq
\int_\Theta \left\| \pd{\b}{\bth} \right\|_F^2 p_\bth(d\bth)
\le \left( \sqrt{t \lambda} + \sqrt{n} \right)^2.
\eeq
Hence, for any $\b \in S_t$,
\begin{align}
\|\b\|^2_{\HoTn} &= \int_\Theta \|\b(\bth)\|^2 p_\bth(d\bth)
                  + \int_\Theta \left\|\pd{\b}{\bth}\right\|_F^2 p_\bth(d\bth) \notag\\
               &\le v + \left( \sqrt{t \lambda} + \sqrt{n} \right)^2.
\end{align}
Thus $S_t$ is bounded for all $t$. It is straightforward to show that $S_t$ is also closed and convex. Therefore, employing Lemma~\ref{le:uniq min gen} (with $\ell=1$) ensures that there exists a (unique) $\bopt \in S_t$ minimizing \eqref{eq:prf high 1}.

Note that the objective in \eqref{eq:prf high 1} is 0 if and only if $\bopt(\bth)={\bf 0}$ almost everywhere. So, if ${\bf 0} \in S_t$, we have $w(t)=0$, and otherwise $w(t)>0$. Let us define
\beq \label{eq:def s}
s \triangleq \E{ \Tr( \J^{-1}(\bth) ) }
\eeq
and note that ${\bf 0} \in S_t$ if and only if $t \ge s$. Thus
\begin{align} \label{eq:prf high 2}
w(t) &= 0 \quad \text{for } t \ge s \notag\\
w(t) &> 0 \quad \text{otherwise.}
\end{align}

Let us now return to the setting of Theorem~\ref{th:high SNR}. For simplicity, we denote functions corresponding to the problem of estimating $\bth$ from $\{ \x^{(1)}, \ldots, \x^{(N)} \}$ with a superscript $(N)$. For example, from the additive property of the Fisher information \cite[\S3.4]{kay93}, we have
\beq
\J^{(N)}(\bth) = N \J(\bth).
\eeq
It follows that
\beq
(N+1)\CRB^{(N+1)}[\b,\bth] \ge N \CRB^{(N)}[\b,\bth]
\eeq
for all $\b\in\HoTn$, all $\bth\in\Theta$, and all $N$. Therefore
\beq
(N+1) Z^{(N+1)}[\b] \ge N Z^{(N)}[\b]
\eeq
for all $\b \in \HoTn$, and hence
\begin{align}
 (N+1) \beta_{N+1}
&=   \min_{\b\in\HoTn} \left( (N+1) Z^{(N+1)}[\b] \right) \notag\\
&\ge \min_{\b\in\HoTn} \left( N Z^{(N)}[\b] \right)       \notag\\
&= N \beta_N.
\end{align}
Thus $\{ N \beta_N \}$ is a non-decreasing sequence. Furthermore, we have
\beq
N Z^{(N)}[{\bf 0}] = s
\eeq
so that $N \beta_N \le s$ for all $N$. It follows that $\{ N \beta_N \}$ is non-decreasing and bounded, and therefore converges to some value $r$ such that
\beq
N \beta_N \le r \le s \quad \text{for all }N.
\eeq
To prove the theorem, we must show that $r=s$.

Let $\b^{(N)} \in \HoTn$ denote the minimizer of \eqref{eq:min HoTn} when $\bth$ is estimated from $\{ \x^{(1)}, \ldots, \x^{(N)} \}$ (the existence of $\b^{(N)}$ is guaranteed by Proposition~\ref{pr:uniq min}). We then have $N \beta_N = N Z^{(N)}[\b^{(N)}] \le r$, so that
\beq
\int_\Theta \Tr\!\left( \! \left(\I+\pd{\b^{(N)}}{\bth}\right) \! \J^{-1}(\bth) \!
                        \left(\I+\pd{\b^{(N)}}{\bth}\right)^T \right) p_\bth(d\bth) \le r.
\eeq
Thus, $\b^{(N)}$ satisfies the constraint of \eqref{eq:prf high 1} with $t=r$. As a consequence, we have
\beq
\int_\Theta \|\b^{(N)}(\bth)\|^2 p_\bth(d\bth) \ge w(r)
\eeq
and therefore
\begin{align} \label{eq:prf high 3}
N \beta_N
&= N Z^{(N)}[\b^{(N)}] \notag\\
&\ge N \int_\Theta \|\b^{(N)}(\bth)\|^2 p_\bth(d\bth) \notag\\
&\ge N w(r).
\end{align}
Now suppose by contradiction that $r<s$. It follows from \eqref{eq:prf high 2} that $w(r)>0$. Hence, by \eqref{eq:prf high 3}, $N \beta_N \rightarrow \infty$, which contradicts the fact that $N \beta_N$ is bounded. We conclude that $r=s$, as required.
\end{proof}

\bibliographystyle{IEEEtran}
\bibliography{IEEEabrv,mybib}

\begin{thebibliography}{10}
\providecommand{\url}[1]{#1}
\csname url@samestyle\endcsname
\providecommand{\newblock}{\relax}
\providecommand{\bibinfo}[2]{#2}
\providecommand{\BIBentrySTDinterwordspacing}{\spaceskip=0pt\relax}
\providecommand{\BIBentryALTinterwordstretchfactor}{4}
\providecommand{\BIBentryALTinterwordspacing}{\spaceskip=\fontdimen2\font plus
\BIBentryALTinterwordstretchfactor\fontdimen3\font minus
  \fontdimen4\font\relax}
\providecommand{\BIBforeignlanguage}[2]{{%
\expandafter\ifx\csname l@#1\endcsname\relax
\typeout{** WARNING: IEEEtran.bst: No hyphenation pattern has been}%
\typeout{** loaded for the language `#1'. Using the pattern for}%
\typeout{** the default language instead.}%
\else
\language=\csname l@#1\endcsname
\fi
#2}}
\providecommand{\BIBdecl}{\relax}
\BIBdecl

\bibitem{berger85}
J.~O. Berger, \emph{Statistical Decision Theory and {Bayesian} Analysis},
  2nd~ed.\hskip 1em plus 0.5em minus 0.4em\relax New York, NY: Springer-Verlag,
  1985.

\bibitem{kay93}
S.~M. Kay, \emph{Fundamentals of Statistical Signal Processing: Estimation
  Theory}.\hskip 1em plus 0.5em minus 0.4em\relax Englewood Cliffs, NJ:
  Prentice Hall, 1993.

\bibitem{vantrees68}
H.~L. Van~Trees, \emph{Detection, Estimation, and Modulation Theory}.\hskip 1em
  plus 0.5em minus 0.4em\relax New York: Wiley, 1968, vol.~1.

\bibitem{ZivZakai69}
J.~Ziv and M.~Zakai, ``Some lower bounds on signal parameter estimation,''
  \emph{{IEEE} Trans. Inf. Theory}, vol.~15, no.~3, pp. 386--391, May 1969.

\bibitem{young71}
T.~Y. Young and R.~A. Westerberg, ``Error bounds for stochastic estimation of
  signal parameters,'' \emph{{IEEE} Trans. Inf. Theory}, vol.~17, no.~5, pp.
  549--557, Sep. 1971.

\bibitem{BelliniTartara74}
S.~Bellini and G.~Tartara, ``Bounds on error in signal parameter estimation,''
  \emph{{IEEE} Trans. Commun.}, vol.~22, no.~3, pp. 340--342, 1974.

\bibitem{ChazanZakaiZiv75}
D.~Chazan, M.~Zakai, and J.~Ziv, ``Improved lower bounds on signal parameter
  estimation,'' \emph{{IEEE} Trans. Inf. Theory}, vol.~21, no.~1, pp. 90--93,
  1975.

\bibitem{BobrovskiZakai76}
B.~Z. Bobrovski and M.~Zakai, ``A lower bound on the estimation error for
  certain diffusion problems,'' \emph{{IEEE} Trans. Inf. Theory}, vol.~22,
  no.~1, pp. 45--52, Jan. 1976.

\bibitem{WeissWeinstein85}
A.~J. Weiss and E.~Weinstein, ``A lower bound on the mean-square error in
  random parameter estimation,'' \emph{{IEEE} Trans. Inf. Theory}, vol.~31,
  no.~5, pp. 680--682, Sep. 1985.

\bibitem{WeinsteinWeiss88}
E.~Weinstein and A.~J. Weiss, ``A general class of lower bounds in parameter
  estimation,'' \emph{{IEEE} Trans. Inf. Theory}, vol.~34, no.~2, pp. 338--342,
  Mar. 1988.

\bibitem{bell97}
K.~L. Bell, Y.~Steinberg, Y.~Ephraim, and H.~L. Van~Trees, ``Extended
  {Ziv}--{Zakai} lower bound for vector parameter estimation,'' \emph{{IEEE}
  Trans. Inf. Theory}, vol.~43, no.~2, pp. 624--637, 1997.

\bibitem{renaux06}
A.~Renaux, P.~Forster, P.~Larzabal, and C.~Richmond, ``The {B}ayesian {A}bel
  bound on the mean square error,'' in \emph{Proc. Int. Conf. Acoust., Speech
  and Signal Processing ({ICASSP} 2006)}, vol. III, Toulouse, France, May 2006,
  pp. 9--12.

\bibitem{lehmann98}
E.~L. Lehmann and G.~Casella, \emph{Theory of Point Estimation}, 2nd~ed.\hskip
  1em plus 0.5em minus 0.4em\relax New York: Springer, 1998.

\bibitem{eldar08b}
Y.~C. Eldar, ``Rethinking biased estimation: Improving maximum likelihood and
  the {C}ram\'er--{R}ao bound,'' \emph{Foundations and Trends in Signal
  Processing}, vol.~1, no.~4, pp. 305--449, 2008.

\bibitem{kay08}
S.~M. Kay and Y.~C. Eldar, ``Rethinking biased estimation,'' \emph{{IEEE}
  Signal Process. Mag.}, vol.~25, no.~3, pp. 133--136, May 2008.

\bibitem{cramer45}
H.~Cram\'er, ``A contribution to the theory of statistical estimation,''
  \emph{Skand. Akt. Tidskr.}, vol.~29, pp. 85--94, 1945.

\bibitem{rao45}
C.~R. Rao, ``Information and accuracy attainable in the estimation of
  statistical parameters,'' \emph{Bull. Calcutta Math. Soc.}, vol.~37, pp.
  81--91, 1945.

\bibitem{hammersley50}
J.~M. Hammersley, ``On estimating restricted parameters,'' \emph{J. Roy.
  Statist. Soc. B}, vol.~12, no.~2, pp. 192--240, 1950.

\bibitem{ChapmanRobbins51}
D.~G. Chapman and H.~Robbins, ``Minimum variance estimation without regularity
  assumptions,'' \emph{Ann. Math. Statist.}, vol.~22, no.~4, pp. 581--586, Dec.
  1951.

\bibitem{bhattacharya66}
P.~K. Bhattacharya, ``Estimating the mean of a multivariate normal population
  with general quadratic loss function,'' \emph{Ann. Math. Statist.}, vol.~37,
  no.~6, pp. 1819--1824, Dec. 1966.

\bibitem{barankin49}
E.~W. Barankin, ``Locally best unbiased estimates,'' \emph{Ann. Math.
  Statist.}, vol.~20, no.~4, pp. 477--501, Dec. 1949.

\bibitem{abel93}
J.~S. Abel, ``A bound on mean-square-estimate error,'' \emph{{IEEE} Trans. Inf.
  Theory}, vol.~39, no.~5, pp. 1675--1680, 1993.

\bibitem{hero96}
A.~O. Hero, J.~A. Fessler, and M.~Usman, ``Exploring estimator bias-variance
  tradeoffs using the uniform {CR} bound,'' \emph{{IEEE} Trans. Signal
  Process.}, vol.~44, no.~8, pp. 2026--2041, 1996.

\bibitem{forster02}
P.~Forster and P.~Larzabal, ``On lower bounds for deterministic parameter
  estimation,'' in \emph{Proc. Int. Conf. Acoust., Speech and Signal Processing
  ({ICASSP} 2002)}, vol.~2, Orlando, FL, May 2002, pp. 1137--1140.

\bibitem{eldar04e}
Y.~C. Eldar, ``Minimum variance in biased estimation: Bounds and asymptotically
  optimal estimators,'' \emph{{IEEE} Trans. Signal Process.}, vol.~52, no.~7,
  pp. 1915--1930, 2004.

\bibitem{eldar06}
------, ``Uniformly improving the {C}ram\'{e}r-{R}ao bound and
  maximum-likelihood estimation,'' \emph{{IEEE} Trans. Signal Process.},
  vol.~54, no.~8, pp. 2943--2956, 2006.

\bibitem{eldar08d}
------, ``{MSE} bounds with affine bias dominating the {C}ram\'er--{R}ao
  bound,'' \emph{{IEEE} Trans. Signal Process.}, vol.~56, no.~8, pp.
  3824--3836, Aug. 2008.

\bibitem{ibragimov81}
I.~A. Ibragimov and R.~Z. Has'minskii, \emph{Statistical Estimation: Asymptotic
  Theory}.\hskip 1em plus 0.5em minus 0.4em\relax New York: Springer, 1981.

\bibitem{renaux_thesis}
\BIBentryALTinterwordspacing
A.~Renaux, ``\BIBforeignlanguage{French}{Contribution \`a l'analyse des
  performances d'estimation en traitement statistique du signal},'' Ph.D.
  dissertation, \`Ecole Normale Superieure de Cachan, 2006. [Online].
  Available: \url{http://tel.archives-ouvertes.fr/tel-00129527/}
\BIBentrySTDinterwordspacing

\bibitem{ben-haim07}
Z.~Ben-Haim and Y.~C. Eldar, ``A {B}ayesian estimation bound based on the
  optimal bias function,'' in \emph{Proc. 2nd Int. Workshop on Computational
  Adv. in Multi-Sensor Adapt. Process. (CAMSAP 2007)}, St. Thomas, U.S. Virgin
  Islands, Dec. 2007.

\bibitem{shao03}
J.~Shao, \emph{Mathematical Statistics}, 2nd~ed.\hskip 1em plus 0.5em minus
  0.4em\relax New York: Springer, 2003.

\bibitem{lieb01}
E.~H. Lieb and M.~Loss, \emph{Analysis}, 2nd~ed.\hskip 1em plus 0.5em minus
  0.4em\relax American Mathematical Society, 2001.

\bibitem{cox87}
D.~R. Cox and N.~Reid, ``Parameter orthogonality and approximate conditional
  inference,'' \emph{J. Roy. Statist. Soc. B}, vol.~49, no.~1, pp. 1--39, 1987.

\bibitem{vantrees07}
H.~L. Van~Trees and K.~L. Bell, \emph{Bayesian Bounds for Parameter Estimation
  and Nonlinear Filtering/Tracking}.\hskip 1em plus 0.5em minus 0.4em\relax New
  York: Wiley, 2007.

\bibitem{vinogradov95}
I.~M. Vinogradov, Ed., \emph{Encyclopaedia of Mathematics}.\hskip 1em plus
  0.5em minus 0.4em\relax Dordrecht, The Netherlands: Kluwer, 1995.

\bibitem{abramowitz64}
M.~Abramowitz and I.~A. Stegun, \emph{Handbook of Mathematical Functions with
  Formulas, Graphs, and Mathematical Tables}.\hskip 1em plus 0.5em minus
  0.4em\relax New York: Dover, 1964.

\bibitem{ben-haim07b}
Z.~Ben-Haim and Y.~C. Eldar, ``A comment on the use of the {W}eiss--{W}einstein
  bound with constrained parameter sets,'' \emph{{IEEE} Trans. Inf. Theory},
  vol.~54, no.~10, pp. 4682--4684, Oct. 2008.

\bibitem{lebedev03}
L.~P. Lebedev and M.~J. Cloud, \emph{The Calculus of Variations and Functional
  Analysis}.\hskip 1em plus 0.5em minus 0.4em\relax New Jersey: World
  Scientific, 2003.

\bibitem{rudin-FunctionalAnalysis}
W.~Rudin, \emph{Functional Analysis}.\hskip 1em plus 0.5em minus 0.4em\relax
  New York: McGraw-Hill, 1973.

\bibitem{gelfand00}
I.~M. Gelfand and S.~V. Fomin, \emph{Calculus of Variations}.\hskip 1em plus
  0.5em minus 0.4em\relax Mineola, NY: Dover, 2000.

\end{thebibliography}

\end{document}